\documentclass[acmtog,nonacm]{acmart}

\usepackage{booktabs} 

\setcitestyle{nosort,square} 

\usepackage{picinpar,moresize,xfrac,graphpap,dcolumn,wrapfig,graphicx}

\microtypesetup{stretch=30,shrink=30}
\DeclareGraphicsExtensions{.png,.pdf,.jpg}
\usepackage{subcaption}
\usepackage{stackengine}

\usepackage{tikz,pgfplots}
\usepackage{tikz-cd}
\usepackage{pgf}
\usepgfplotslibrary{colormaps}
\usetikzlibrary{pgfplots.colormaps}
\usetikzlibrary{arrows,automata,positioning,backgrounds}
\usepackage{rotating}
\pgfplotsset{compat=newest}
\pgfplotsset{plot coordinates/math parser=false}
\newlength\figureheight
\newlength\figurewidth

\usepackage{soul,array,calc,url,ragged2e,graphpap}
\usepackage{booktabs} 
\usepackage{tabularx}
\usepackage{colortbl}
\usepackage{multirow}
\usepackage{hyphenat}
\usepackage{transparent}
\usepackage{enumerate}
\usepackage{mathrsfs}
\usepackage{mdframed}
\usepackage{acro}

\usepackage[figure]{hypcap}
\usepackage{mathtools}
\mathtoolsset{centercolon}
\usepackage{nicefrac}
\usepackage{units}
\usepackage[normalem]{ulem}
\usepackage{cancel}
\usepackage{pbox}
\usepackage{multicol}
\usepackage{cases}
\usepackage[all]{xy}
\usepackage{nicematrix}

\usepackage{adjustbox}
\usepackage{wrapfig}
\AfterEndEnvironment{wrapfigure}{\setlength{\intextsep}{0mm}}

\usepackage{cleveref}
\crefmultiformat{subequation}%
{\edef\crefstripprefixinfo{#1}(#2#1#3}%
{,#2\crefstripprefix{\crefstripprefixinfo}{#1}#3)}%
{,#2\crefstripprefix{\crefstripprefixinfo}{#1}#3}%
{,#2\crefstripprefix{\crefstripprefixinfo}{#1}#3)}

\newcommand{\overbar}[1]{\mkern 1.5mu\overline{\mkern-1.5mu#1\mkern-1.5mu}\mkern 1.5mu}
\newcommand*{\conj}[1]{\overbar{#1}}

\makeatletter
\DeclareFontFamily{OMX}{MnSymbolE}{}
\DeclareSymbolFont{MnLargeSymbols}{OMX}{MnSymbolE}{m}{n}
\SetSymbolFont{MnLargeSymbols}{bold}{OMX}{MnSymbolE}{b}{n}
\DeclareFontShape{OMX}{MnSymbolE}{m}{n}{
    <-6>  MnSymbolE5
   <6-7>  MnSymbolE6
   <7-8>  MnSymbolE7
   <8-9>  MnSymbolE8
   <9-10> MnSymbolE9
  <10-12> MnSymbolE10
  <12->   MnSymbolE12
}{}
\DeclareFontShape{OMX}{MnSymbolE}{b}{n}{
    <-6>  MnSymbolE-Bold5
   <6-7>  MnSymbolE-Bold6
   <7-8>  MnSymbolE-Bold7
   <8-9>  MnSymbolE-Bold8
   <9-10> MnSymbolE-Bold9
  <10-12> MnSymbolE-Bold10
  <12->   MnSymbolE-Bold12
}{}
\let\llangle\@undefined
\let\rrangle\@undefined
\DeclareMathDelimiter{\llangle}{\mathopen}%
                     {MnLargeSymbols}{'164}{MnLargeSymbols}{'164}
\DeclareMathDelimiter{\rrangle}{\mathclose}%
                     {MnLargeSymbols}{'171}{MnLargeSymbols}{'171}
\makeatother

\usepackage{fancyvrb,listings}
\usepackage{algorithm}
\usepackage{algorithmicx}
\usepackage[noend]{algpseudocode}

\algrenewcommand\alglinenumber[1]{\footnotesize #1:}
\makeatletter  
 \renewcommand{\ALG@name}{\small Algorithm} 
\makeatother 
\algdef{SE}[DOWHILE]{Do}{doWhile}{\algorithmicdo}[1]{\algorithmicwhile\ #1}%

\newtheoremstyle{mine}{3pt}{3pt}{\itshape}{}{\bfseries}{.}{.5em}{}
\theoremstyle{mine}
\newtheorem{theorem}{Theorem}[section]
\newtheorem{lemma}{Lemma}[section]
\newtheorem{proposition}{Proposition}[section]
\newtheorem{corollary}{Corollary}[section]
\newtheorem{definition}{Definition}[section]

\newtheorem{problem}{Problem}[section]

\def\ie{\emph{i.e.}}
\def\eg{\emph{e.g.}}

\def\CC{\mathbb{C}}
\def\DD{\mathbb{D}}

\def\HH{\mathbb{H}}

\def\NN{\mathbb{N}}

\def\RR{\mathbb{R}}
\def\SS{\mathbb{S}}

\def\ZZ{\mathbb{Z}}

\def\bB{\mathbf{B}}

\def\bN{\mathbf{N}}

\def\bQ{\mathbf{Q}}

\def\bT{\mathbf{T}}

\def\cL{\mathcal{L}}

\def\bc{\mathbf{c}}

\def\bn{\mathbf{n}}

\def\br{\mathbf{r}}

\def\bv{\mathbf{v}}
\def\bw{\mathbf{w}}
\def\bx{\mathbf{x}}
\def\by{\mathbf{y}}

\def\brho{\boldsymbol{\rho}}

\def\bphi{\boldsymbol{\phi}}

\def\bpsi{\boldsymbol{\psi}}

\def\bzero{\mathbf{0}}

\usepackage{bbm}
\DeclareSymbolFont{bbold}{U}{bbold}{m}{n}
\DeclareSymbolFontAlphabet{\mathbbold}{bbold}
\newcommand{\ii}{\mkern1.5mu\mathbbold{i}\mkern1.5mu}

\renewcommand{\Im}{\operatorname{Im}}
\renewcommand{\Re}{\operatorname{Re}}

\DeclareMathOperator{\SO}{SO}

\DeclareMathOperator{\SU}{SU}

\DeclareMathOperator{\Homeo}{Homeo}

\DeclarePairedDelimiterX\braket[2]{\langle}{\rangle}{#1\,\delimsize\vert\,\mathopen{}#2}

\def\bpsi{\boldsymbol{\psi}}

\definecolor{b1}{rgb}{0.158099,0.313781,0.636957}
\definecolor{b2}{rgb}{0.525367,0.691857,0.998936}
\definecolor{g1}{rgb}{0.256000,0.640000,0.576000}
\definecolor{g2}{rgb}{0.559573,0.800781,0.760580}
\definecolor{p1}{rgb}{0.416000,0.192000,0.640000}
\definecolor{p2}{rgb}{0.680880,0.561880,0.799881}
\definecolor{r1}{rgb}{0.800000,0.200000,0.000000}
\definecolor{r2}{rgb}{1.000000,0.702595,0.603461}
\definecolor{k1}{gray}{0}
\definecolor{k2}{gray}{0.7}

\usepackage{enumitem}
\setcopyright{rightsretained}
\acmJournal{TOG}
\acmYear{2026} \acmVolume{0} \acmNumber{0} \acmArticle{0} \acmMonth{0}\acmDOI{0}
\makeatletter
\let\@authorsaddresses\@empty
\makeatother
\begin{document}
\title{Neural Representation of Minimal Surfaces}

\author{Jiayin Sun}
 \affiliation{
   \institution{University of Utah}
   \streetaddress{201 Presidents Circle}
   \city{Salt Lake City}
   \state{UT}
   \postcode{84112}
   \country{USA}
 }
 \email{alchern@ucsd.edu}
\author{Albert Chern}
 \affiliation{
   \institution{University of California San Diego}
   \streetaddress{9500 Gilman Dr, MC 0404}
   \city{La Jolla}
   \state{CA}
   \postcode{92093}
   \country{USA}
 }
 \email{alchern@ucsd.edu}

\newcommand{\AC}[1]{{\footnotesize\color[rgb]{0.2,0.4,0.8}\textbf{AC:} #1}}
\newcommand{\JS}[1]{{\footnotesize\color[rgb]{0.2,0.8,0.4}\textbf{JS} #1}}

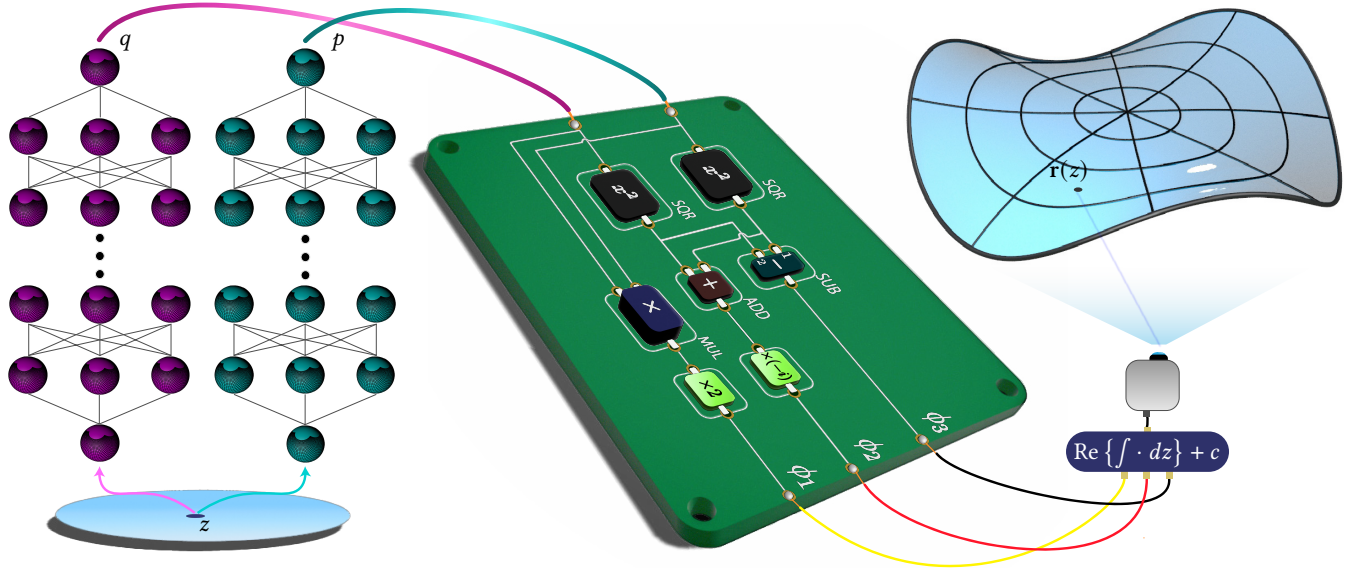
\begin{teaserfigure}
    \centering
    \setlength{\unitlength}{1pt}
    \begin{picture}(500,220)
        \put(0,0){
        \includegraphics[width=\linewidth]{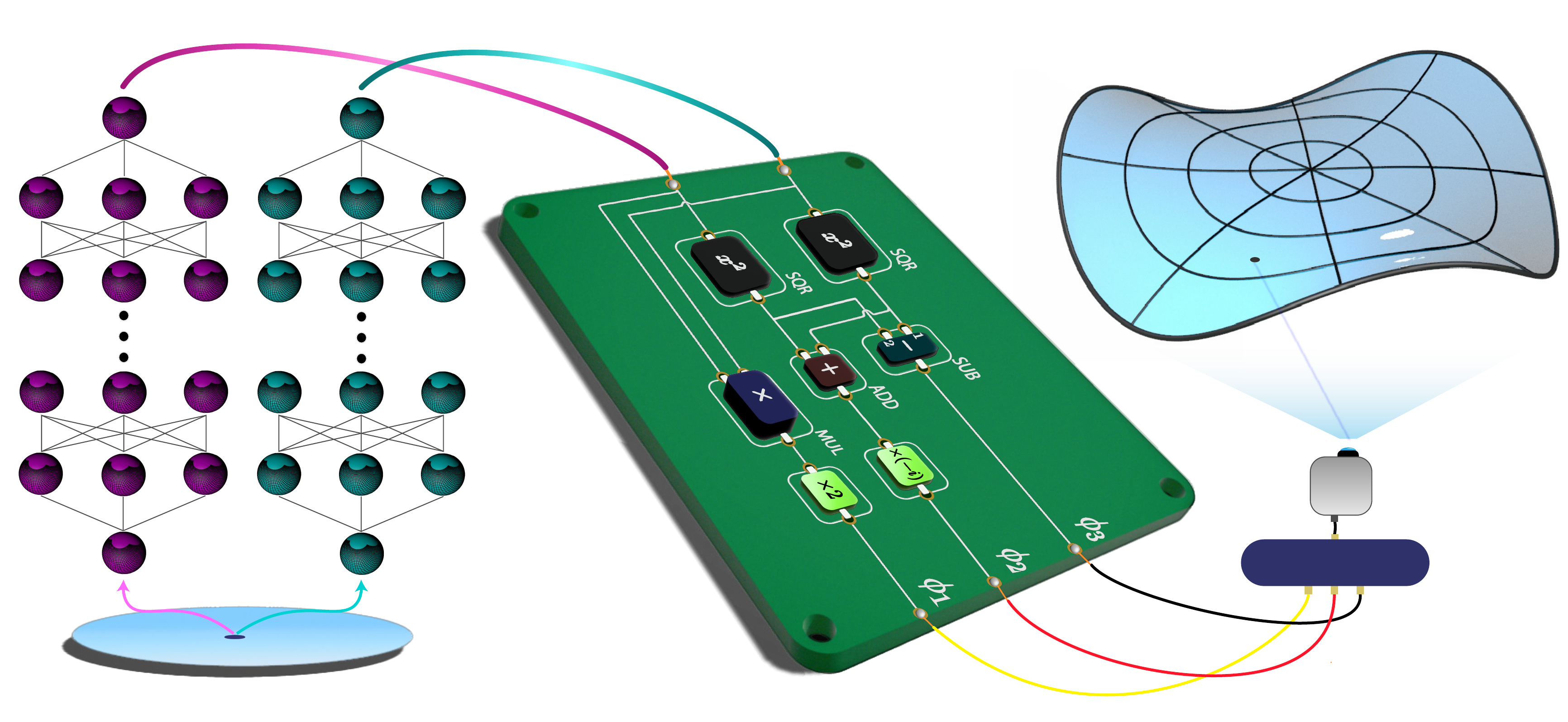}
            }
        \put(50,205){$q$}
        \put(130,205){$p$}
        \put(80,22){$z$}
        \put(400,155){$\br(z)$}
        \put(407,49){ \textcolor{white}{${ \operatorname{Re} \left\{ \int \cdot  \;dz \right\} + c}$}}
    \end{picture}
    \caption{Overview of the Neural Architecture. The holomorphic functions $p$ and $q$ are constructed as complex-valued neural network (CVNN) whose inputs, weights and activations take values in $\mathbb{C}$ (illustrated compactly as Riemann spheres). The output of $p$ and $q$ are then encoded into the holomorphic triple $(\phi_1, \phi_2, \phi_3)$ by Euclid's spinor representation \eqref{eq:Pythagorean}, and the integration of such triple intrinsically yields a minimal surface in $\mathbb{R}^3$.
    }
    \label{fig:NN_diagram}
\end{teaserfigure}

\begin{abstract}

We propose a neural representation for minimal surfaces. Unlike prior approaches based on discretization or Physics-Informed Neural Networks (PINNs), where meshes or neural fields are optimized to approximate the governing equations, our method builds on an exact representation, similar to the classical Weierstrass--Enneper parameterization,
yielding minimal surfaces up to negligible quadrature error in evaluation.
We formulate a training objective for the Plateau problem that optimizes over this representation.
\end{abstract}

\ccsdesc[500]{Computing methodologies~Parametric curve and surface models}
\ccsdesc[300]{Computing methodologies~Shape analysis}
\ccsdesc[300]{Computing methodologies~Neural networks}

\keywords{Minimal surfaces}

\maketitle



\section{Introduction}
A \emph{minimal surface} is a surface with vanishing mean curvature everywhere.  In physical terms, these surfaces have no net surface tension at each point, and thus they describe the shapes of soap films in equilibrium.  
The construction of minimal surfaces is one of the central interests in geometry processing \cite{Wang:2021:CMS}, discrete differential geometry \cite{pinkall1993computing,bobenko2006minimal}, architectural geometry \cite{tenu2010minimal,velimirovic2008minimal}, material science \cite{al2019multifunctional,reynolds2023characterisation}, and mathematical visualization \cite{melko2010visualizing,palais1999visualization}. 
The classic problem for minimal surfaces is the \emph{Plateau problem}: given a closed space curve, find a minimal surface bounded by it. 

After general mathematical existence results were established for the Plateau problem, many methods were developed for solving it numerically.
A common method is by discretizing the surface into a mesh and solve for the discrete minimal surface subject to the given boundary condition \cite{pinkall1993computing}. A few recent methods taking the geometric measure theoretic route solve for the \emph{currents} representing the minimal surface over the ambient 3-space \cite{parks1997computing,Wang:2021:CMS}.
A caveat of these numerical methods are that they do not generate exact minimal surfaces in the presence of discretization.
One may avoid discretization by employing a \emph{neural field} representation for the unknown embedding function or currents \cite{Palmer_2022_CVPR,berzins2024geometry}.  However, in this ``physics-informed neural network'' strategy, the neural fields are being optimized over finite sample points effectively giving rise to an implicit discretization resolution.  Afterall, minimal surfaces are solutions to a nonlinear PDE, for which general computable representation are not exact.



It turns out that there is a way to represent minimal surface \emph{exactly}. It is known as the \emph{Weierstrass--Enneper parameterization} \cite{weierstrass1866untersuchungen,enneper1868analytisch,osserman2013survey}. Given any holomorphic function $f(\zeta)$ and a meromorphic function $g(\zeta)$ defined on a domain \(\Omega\subset\CC\) of the complex plane, the surface \(\br\colon\Omega\to \RR^3\) defined by
\begin{equation}
\label{eq:WE}
\textstyle
\br(z)
\coloneqq
\operatorname{Re}
\int_{0}^{z}
\left(
\frac{1}{2}f(1-g^2),
\frac{\ii}{2}f(1+g^2),
fg
\right)
\, d\zeta
\end{equation}
is \emph{guaranteed} to be a minimal surface.
In the computational context, it is natural to ``neuralize'' this Weierstrass--Enneper parameterization, since it is straightforward to construct \emph{Complex-Valued Neural Networks} (CVNNs) \cite{lee2022complex,calafa2024physics} that represent functions that are exactly holomorphic.
There are also Universal Approximation Theorems for nueral representations of holomorphic functions \cite{calafa2024physics}.
Eq.~\eqref{eq:WE} can be interpreted as an ``architecture'' for combining neural networks \(f\) and \(g\) to represent a surface.  This architecture is designed with such a strong prior for the minimal surfaces problem that any surface \(\br\) represented by it is an exact minimal surface.

While the Weierstrass--Enneper formula appears powerful, there is an unavoidable problem if we directly replace \(f,g\) by neural networks.
The meromorphic function $g$ has a geometric meaning: \(g\) is the \emph{Gauss map} (normal) of the surface, after a post-composition with the stereographic projection from \(\CC\cup\{\infty\}\) to the unit sphere.  
In particular, a pole inevitably appears in \(g\) whenever the surface normal points to the ``up'' direction of the given \(\RR^3\) coordinate system.  These pole singularities in meromorphic functions are difficult to represent by complex neural network, unlike holomorphic functions with the universal approximation property.
Another subtle problem of \eqref{eq:WE} is that the formula can represent surfaces with \emph{branch points} that are undesirable in the solutions of the Plateau problem.

To address these problems, we study the fundamental characterization for a surface to be minimal, and propose an alternative formula of \eqref{eq:WE}. 
This formula is the lesser known \emph{spinor representation of minimal surfaces} \cite{kusner1995spinor}. It only requires \emph{two holomorphic functions}, rather than having one of them being meromorphic.  
We show that any solution to the Plateau problem for rectifiable Jordan curves can be parameterized by the spinor formula.  We also show that every spinor represented minimal surface is either regular or has only removable branch points.  
Using this formula, we represent minimal surfaces with complex-valued neural networks (\Cref{fig:NN_diagram}).  We design training tasks that find the two holomorphic functions that solve the Plateau problem for a given curve.



\subsection{Related Work}

 Alongside the classical analytic formulations of the Plateau problem  \cite{lagrange1761essai,douglas1931solution,osserman1970proof}, computational approaches for constructing minimal surfaces have also been extensively studied over a century. 
Douglas~\shortcite{douglas1927method} discretized the surface as a height function over a grid. 
This discrete approximation was approached with root finding algorithms  \cite{concus1967numerical} and finite element methods \cite{hinata1974numerical}. 
Beyond height-field representations, more general minimal surfaces were obtained by representing the surface with triangle meshes and minimizing either the Dirichlet energy \cite{wilson1961discrete} or the area functional \cite{wagner1977contribution}.
Modern methods include mean curvature flow \cite{dziuk1990algorithm,brakke1992surface}, $H^1$-regularized curvature flow \cite{pinkall1993computing,schumacher2019variational} and conformal Willmore flow \cite{crane2011spin}.
Minimal surfaces also arise as \emph{minimal currents} in the geometric measure theoretic formulation \cite{federer1960normal,almgren2006co,brezis2019plateau}.
This formulation casts the Plateau problem into an \(L^1\) minimization for a 1-form in the ambient 3D space \cite{parks1997computing,Wang:2021:CMS,Palmer_2022_CVPR}.  Related discrete analogues, known as the \emph{optimal homologous chain problem}, was studied in \cite{sullivan1990crystalline,dunfield2011least,dey2010optimal,cohen2022lexicographic}.

Recent advances in neural networks have enabled continuous surface representations that avoid fixed-resolution discretizations.
Implicit neural representations model surfaces using neural fields that approximate a signed distance field \cite{park2019deepsdf,schirmer2024geometric}, an occupancy field \cite{mescheder2019occupancy}, or a current \cite{Palmer_2022_CVPR}.  
Neural networks have also been used to represent explicit surface parameterizations \cite{morreale2021neural,groueix2018papier,sivaram2024neural}.
Recent works have further explored the use of neural networks for constructing minimal surfaces. 
\cite{zhou2023approximating,kabasi2023physics} represent surfaces as height fields and employ Physics-Informed Neural Networks (PINNs) to solve the Euler-Lagrange equation of the area functional under prescribed boundary conditions. Similarly, Hashimoto et al.~\shortcite{hashimoto2026physics} use a PINN to solve the Euler-Lagrange equations for constructing minimal surfaces in curved spacetimes.
Berzins et al.~\shortcite{berzins2024geometry} introduced Geometry-Informed Neural Networks (GINNs), in which implicit surface representations are optimized with geometric functionals and boundary conditions to solve the Plateau problem and other geometric optimization problems.

In contrast to these approaches, where minimality is achieved by optimizing the network to satisfy the governing PDEs, we adopt a surface representation that intrinsically encodes the minimality condition.  Our approach is conceptually similar to that of \cite{calafa2024physics} for two-dimensional linear elasticity, where the general solution admits a holomorphic representation that is subsequently neuralized.

\section{Minimal Surface Representations}


In this section, we introduce our neural representation for minimal surfaces. 
The design is based on a modification of the Weierstrass--Enneper parameterization. 
We begin by reviewing the classical theory of holomorphic representations of minimal surfaces, and then present our representation along with a neural architecture built upon it.

Throughout this paper, we consider surfaces with disk topology. Let \(\DD\subset\RR^2\cong\CC\) denote the (open) unit disk.  
A \emph{surface} is any smooth map \(\br\colon \DD\to\RR^3\).
A surface is \emph{regular} (or an \emph{immersion}) if \(d\br|_p\colon T_p\DD\to\RR^3\) has rank 2 for all \(p\in \DD\).

A surface is said to be \emph{conformal} if \(|{\partial \br\over\partial x}| = |{\partial \br\over\partial y}|\) and \(\langle{\partial \br\over\partial x},{\partial\br\over\partial y}\rangle = 0\). 
A surface is \emph{harmonic} if \(\Delta\br = {\partial^2\br\over\partial x^2}+ {\partial^2\br\over\partial y^2}=\bzero\).%
\footnote{
Note that the notions of conformality and harmonicity depend only on a \emph{conformal structure}, rather than on a specific coordinate system. For an oriented two-dimensional manifold \(M\) endowed with a conformal structure represented by a Hodge star operator \(\star_1\colon T^*M\to T^*M\) acting on 1-forms, a surface \(\br\colon M\to\RR^3\) is conformal if \(|d\br(X)| = |(\star_1 d\br)(X)|\) and \(\langle d\br(X),(\star_1 d\br)(X)\rangle = 0\) for all \(X\in TM\).  A surface is harmonic if \(d\star_1 d\br = \bzero\).
}

\subsection{Minimal Surfaces}

\begin{definition}[Generalized minimal surface]\label{def:GeneralizedMinimalSurface}
    A surface \(\br\colon\DD\to\RR^3\) is a \textbf{generalized minimal surface} if it is both conformal and harmonic.
    A surface is a \textbf{regular minimal surface} if it is regular and a generalized minimal surface.
\end{definition}

    


This definition of minimal surface is equivalent to the standard characterization of vanishing mean curvature.

\begin{proposition}[\cite{osserman2013survey}]
    If $\br$ is a regular minimal surface, then $\br$ has mean curvature $H=0$ everywhere.
\end{proposition}

In general, a generalized minimal surface can feature isolated non-immersed \emph{branch points}, while the rest of the surface is regular with vanishing mean curvature.

\subsubsection{Plateau's Problem}

A class fundamental questions concerning constructing minimal surfaces is the
\emph{Plateau problem}.  They take the general form of ``Extend a given boundary space curve into a minimal surface.'' 

\begin{problem}[Plateau problem] \label{prob:plt_prob}
Given a rectifiable Jordan curve $\Gamma$ in $\mathbb{R}^3$, construct a map $\br: \overline{\DD} \to \RR^3$ so that $\br|_{\partial\DD}$ is an orientation-preserving homeomorphism from \(\partial\DD\) to \(\Gamma\), and $\br|_{\DD}$ is a minimal surface.
\end{problem}
 
\cite{douglas1931solution} and \cite{osserman1970proof} show that there always exists a regular minimal surface solving \Cref{prob:plt_prob}.





\subsubsection{Null Holomorphic Curves}
\label{sec:NullHolomorphicCurves}
The characterization of minimal surfaces in \Cref{def:GeneralizedMinimalSurface} in terms of harmonicity and conformality makes complex analysis particularly natural to employ.  
Holomorphic functions, viewed as maps between complex planes, are conformal.
The real part of any holomorphic function is harmonic;  conversely, every harmonic function is the real part of a holomorphic function.

Identify the \((x,y)\) coordinates of \(\DD\) with a single complex variable \(z = x+\ii y\). Each harmonic surface \(\br\colon\DD\to\RR^3\) can then be identified as the real part of a triple of holomorphic functions \(\brho = (\rho_1,\rho_2,\rho_3)\colon\DD\to\CC^3\), that is, \(\br(z) = \Re(\brho(z))\).

Since \(\brho\) is a holomorphic map from one complex variable to three complex coordinates, it is also called a \emph{holomorphic curve}.

For each holomorphic curve \(\brho\colon\DD\to\CC^3\), consider its complex derivative \(\bphi\coloneqq {\partial\over\partial z}\brho\colon\DD\to\CC^3\), which is also holomorphic.
Analogous to classical curve theory, \(\bphi\) represents the \emph{complex velocity vector} of the curve \(\brho\).
One can recover the holomorphic curve, and hence the harmonic surface \(\br\), by integrating the velocity \(\bphi\):
\begin{align}\label{eq:HarmonicSurfaceFromNullVelocity}
\textstyle
    r_k(z) = \Re\left(\int_0^z \phi_k(\zeta)\, d\zeta\right)+c_k,\quad z\in\DD,\quad k=1,2,3,
\end{align}
where \(c_k\in\RR\) are constants.

One can verify that the harmonic surface \(\br\) is regular at \(z\in\DD\) if and only if \(\bphi(z)\neq 0\) (\ie\@ \(\sum_{k=1}^3|\phi_k(z)|^2 \neq 0\)).
A less obvious, but still easily verified (Appendix~\ref{appendix:facts_about_harmonic_map}), fact is that \(\br\) is conformal if and only if 
\begin{align}\label{eq:NullCondition}
    \phi_1^2(z)+\phi_2^2(z) + \phi_3^2(z) = 0\quad\text{for all \(z\in\DD\).}
\end{align}
A velocity vector \(\bphi(z)\) is said to be \emph{null} if it satisfies \eqref{eq:NullCondition}.  A holomorphic curve is called a \emph{null holomorphic curve} if its velocity function \(\bphi\) is null everywhere.

In summary, every generalized minimal surface is the real part of a null holomorphic curve.  The surface is, in addition, regular if the velocity of the null curve is nonvanishing.

\subsubsection{Branch Points}

While representing minimal surfaces by null holomorphic curves is straightforward, the integrated surface \(\br\) may exhibit non-immersed branch points. Here, we briefly digress to discuss branched surfaces.

Isolated non-immersed points on a surface can be classified as either \emph{genuine branch points} or \emph{removable} ones \cite[\S~2.2]{chern2018shape}. Genuine branch points are also referred to as \emph{pinch points} in \cite{chern2018shape}.  
In the following, we provide a concrete characterization based on the theory in \cite{chern2018shape} and the related classification results of \cite{kauffman1977immersions,pinkall1984regular,hass1985immersions}.

Let \(\br\colon\DD\to\RR^3\) be a surface with isolated branch points, and let \(\DD' = \{p\in\DD \mid \br \text{ is regular at } p\}\) denote the punctured domain on which \(\br\) is an immersion.  
Using the immersion \(\br|_{\DD'}\), each closed regular curve \(\gamma\colon\SS^1\to\DD'\) can be canonically assigned a frame \(\bQ_\gamma = [\bT,\bN,\bB]\colon\SS^1\to\SO(3)\), defined as the tangent vector \(\bT \coloneqq \textsc{Normalize}(d\br \circ {d\gamma\over dt})\), the surface normal \(\bN \coloneqq \bn \circ \gamma\) (where \(\bn\colon\DD'\to\SS^2\) is the unit normal of the immersion \(\br\)), and \(\bB = \bT \times \bN\).  
In particular, \(\bQ_\gamma\colon\SS^1\to\SO(3)\) defines a loop in \(\SO(3)\).

From the covering space theory, every loop \(\bQ\colon\SS^1\to\SO(3)\) is either liftable or non-liftable to the universal cover \(\SU(2)\) of \(\SO(3)\). This defines a \(\ZZ_2\)-valued function \({\frak q}\) on the space of loops in \(\DD'\):
\begin{align}\label{eq:Z2ValuedQuadraticForm}
    {\frak q}(\gamma) = 
    \begin{cases}
        0, & \text{if \(\bQ_\gamma\) is not liftable to \(\SU(2)\)},\\ 
        1, & \text{if \(\bQ_\gamma\) is liftable to \(\SU(2)\)}.
    \end{cases}
\end{align}
The value \({\frak q}(\gamma)\) depends only on the homology class of \(\gamma\), and \({\frak q}\colon H_1(M;\ZZ)\to \ZZ_2\) is a quadratic form (\ie\@ it satisfies
\({\frak q}(\gamma_1+\gamma_2) = {\frak q}(\gamma_1)+{\frak q}(\gamma_2) + \#(\gamma_1\cap\gamma_2)\) mod \(2\))
\cite{kauffman1977immersions}.

This invariant \({\frak q}(\gamma)\) classifies a branch point via any loop \(\gamma\) enclosing it, as follows. 

Let \(z_0\in\DD\) be an isolated branch point of \(\br\), and let \(B\subset\DD\) be a small disk centered at \(z_0\) containing no other branch points. Then \emph{the immersion \(\br|_{\DD'\setminus B}\) extends to an immersion on \(\DD'\cup \{z_0\}\) if and only if \({\frak q}(\partial B)=0\)}  \cite[\S~2.2]{chern2018shape}.  
In particular, there exists a perturbation \(\mathring\br\), supported in \(B\), such that \(\br + \epsilon \mathring\br\) is an immersion on \(\DD'\cup\{z_0\}\) if and only if \({\frak q}(\partial B)=0\).

Accordingly, branch points with \({\frak q}(\partial B)=0\) are \emph{removable} under perturbation. In contrast, if \({\frak q}(\partial B)=1\), the branch point persists under any local perturbation; such points are \emph{genuine branch points}.%
\footnote{
    This notion of removability under perturbation differs from the notion of \emph{false branch points} introduced by Osserman. A \emph{false branch point} is a non-immersed point that becomes immersed after a reparameterization. In this case, the surface is locally embedded near the point without any perturbation. In contrast, removable branch points in our sense become locally embedded only after a perturbation. False branch points are technically immersed point, whereas the ability to remove a non-immersed point by perturbation is not captured by this true/false classification.}

For a generalized minimal surface \(\br\) defined by \eqref{eq:HarmonicSurfaceFromNullVelocity}, each branch point corresponds to a zero of the null holomorphic velocity \(\bphi\). The following theorem relates the type of branch point to the order of this zero.  
Here, the order \(m\in\NN\) of a zero \(z_0\) of \(\bphi\) is defined so that one can write \(\bphi(z) = (z-z_0)^m \bpsi(z)\) for some \(\bpsi\) is holomorphic and \(\bpsi(z_0)\neq 0\).
\begin{theorem}\label{thm:RemovableBranch}
    Let \(z_0\) be a zero of a null holomorphic triplet \(\bphi\) of order \(m\in\NN\). Then the corresponding branch point of the generalized minimal surface \(\br\), defined by \eqref{eq:HarmonicSurfaceFromNullVelocity}, is removable if and only if \(m\) is even.
\end{theorem}
\begin{proof}
See Appendix~\ref{app:ProofOfRemovableBranch}.
\end{proof}




\subsection{Weierstrass--Enneper Parameterization}
\label{sec:WeierstrassEnneper}
Recall \Cref{sec:NullHolomorphicCurves} that
constructing a minimal surface amounts to finding three holomorphic functions \(\phi_1,\phi_2,\phi_3\) that satisfy the null condition \eqref{eq:NullCondition}. 
This condition allows us to eliminate one of the three functions.
For example, \(\bphi\) is guaranteed to be null if 
\begin{align}\label{eq:PhiPropToExpressOfG}
    \bphi\propto \textstyle({1\over 2}({1-g^2}), {\ii\over 2}({1+g^2}), g)
\end{align}
for an arbitrary function \(g\colon\DD\to\CC\). 
One may recognize \eqref{eq:PhiPropToExpressOfG} as the stereographic projection from the line to the complex conic obtained by projectivizing \eqref{eq:NullCondition} expressed in homogeneous coordinates.  

Introducing another arbitrary function \(f\colon\DD\to\CC\) to account for the proportionality factor, we obtain
\begin{align}\label{eq:WEforPhi}
    \phi_1 = 
    \textstyle
    {1\over 2}f(1-g^2),\quad\phi_2 = {\ii\over 2}f(1+g^2),\quad \phi_3 = fg.
\end{align}
Substituting \eqref{eq:WEforPhi} into \eqref{eq:HarmonicSurfaceFromNullVelocity} yields the Weierstrass--Enneper formula \eqref{eq:WE}.

Note that \(f\) and \(g\), expressed in terms of \(\bphi\), are given by
\begin{align}\label{eq:WEfg}
    f = \phi_1 - \ii\phi_2,\quad g = {\phi_3/(\phi_1-\ii\phi_2)}.
\end{align}
In particular, for a holomorphic curve \(\bphi\), the function \(f\) is holomorphic, whereas \(g\) is generally meromorphic (\ie\@ it may have pole singularities) due to the quotient.

In fact, the meromorphic function \(g\colon\DD\to\CC\cup\{\infty\}\) is the Gauss map of \(\br\), after identifying \(\CC\cup\{\infty\}\) with the Riemann sphere via the standard stereographic projection.  A pole in \(g\) occurs whenever the normal vector of \(\br\) points in the direction \((0,0,1)\).  In particular, pole singularities in \(g\) are inevitable depending on how the surface \(\br\) is oriented in \(\RR^3\).
While holomorphic functions can be constructed directly, constructing meromorphic functions is computationally more challenging.  It often requires prescribing the locations and orders of poles, which introduces additional complexity. 

Another challenge arising in generating minimal surfaces using \eqref{eq:WEforPhi} is that a holomorphic function \(f\) generically has zeros in \(\DD\).  Each simple zero of \(f\) that is not a pole of \(g\) corresponds to a simple zero of \(\bphi\), which leads to a non-removable branch point of \(\br\) by \Cref{thm:RemovableBranch}.  These genuine branch points should not be present in solutions of the Plateau problem (\Cref{prob:plt_prob}).  

Hence, for generating regular minimal surfaces, not only that one needs to handle poles in \(g\), one also need to enforce constraints that \(f\) has zeros only at the poles of \(g\) and nowhere else in the domain.

Next, we propose an alternative to the Weierstrass--Enneper stereographic formula \eqref{eq:PhiPropToExpressOfG} for enforcing the null condition. 
The alternative formulation avoids the need to explicitly handle poles, and the minimal surfaces it represent are automatically regular or feature at most removable branch points.

\subsection{Euclid's Spinor Representation} 

Our parameterization is motivated by \emph{Euclid's formula} for the \emph{Pythagorean triples}.  The classical Euclid formula is a parameterization of the solutions to \(a^2 + b^2 = c^2\) via \(a=m^2-n^2\), \(b=2mn\), \(c=m^2+n^2\).
Motivated by this formula, we represent null triples \((\phi_1,\phi_2,\phi_3)\) (satisfying \eqref{eq:NullCondition}) as
\begin{align}\label{eq:Pythagorean}
    \phi_1 = p^2-q^2,\quad \phi_2 = -\ii(p^2 + q^2),\quad \phi_3 = 2pq.
\end{align}
This parameterization of null triples is known as the null--spinor correspondence \cite{budinich1986null}.  
    Parameterization for minimal surface using this null--spinor correspondence is called the \emph{spinor representation for minimal surfaces} \cite{kusner1995spinor}.

Unlike the parameters \(f\) and \(g\) in \eqref{eq:WEfg}, which may exhibit singularities, the following theorem shows that \(p,q\) are holomorphic.


\begin{lemma}\label{lem:even_order}
Let $h$ be a holomorphic function on \(\DD\). Then there exists a holomorphic function \(u\) on \(\DD\) such that \(u^2 = h\) if and only if every zero of \(h\) has even order.  The function \(u\) is unique up to a sign.
\end{lemma}
\begin{proof}
See Appendix~\ref{app:ProofOfEvenOrderLemma}.
\end{proof}

\begin{theorem}\label{thm:ExistenceOfPythagoreanTriple}
    Let \((\phi_1,\phi_2,\phi_3)\) be a null holomorphic triple on \(\DD\) representing a regular minimal surface (via \eqref{eq:HarmonicSurfaceFromNullVelocity}).  Then there exist holomorphic functions \(p,q\) on \(\DD\) such that \eqref{eq:Pythagorean} holds.
\end{theorem}

\begin{proof}
    Let \(\bphi = (\phi_1,\phi_2,\phi_3)\) be a non-vanishing null holomoprhic triple.  Let \(f = \phi_1-\ii\phi_2\) and \(g = \phi_3/(\phi_1-\ii\phi_2)\) be Weierstrass--Enneper's functions \eqref{eq:WEfg}.  Then \(f\) vanishes only at the poles of \(g\).  Moreover, from the expression \eqref{eq:WEforPhi} and that \(\bphi\) is holomorphic and nonvanishing, the order of any zero of \(f\) must be exactly twice the order of the pole of \(g\) at that point.
    Now consider the functions \(\phi_1+\ii\phi_2\) and \(\phi_1 - \ii\phi_2\), which are expressed in terms of \(f\) and \(g\) as
    \begin{align}
        \phi_1 + \ii\phi_2 = -fg^2,\quad \phi_1 - \ii\phi_2 = f.
    \end{align}
    Observe that if \(\phi_1 + \ii\phi_2 = -fg^2\) has a zero at \(z\), then \(g\) must have a zero and \(f\) must not have a zero at \(z\).  In particular, the zeros of \(fg^2\) must have even order.  
    By \Cref{lem:even_order}, there exists a holomorphic function \(p\) on \(\DD\) such that \(p^2 = {\phi_1 + \ii\phi_2\over 2}\).  Similarly, since the orders of zeros of \(f\) are twice of the orders of poles of \(g\), the holomorphic function \(f\) only have zeros of even order.  Thus there exists a holomorphic function \(q\) on \(\DD\) such that \(q^2 = -{\phi_1 - \ii\phi_2\over 2}\).
    Now we have
    \begin{align}\label{eq:Phi1Phi2ByPQ}
        \phi_1 = p^2 - q^2,\quad \phi_2 = -\ii(p^2 + q^2).
    \end{align}
    It remains to check whether \(\phi_3 = 2pq\).  Note that \((2pq)^2 = 4p^2q^2 = 4({\phi_1 + \ii\phi_2\over 2})({\phi_1-\ii\phi_2\over 2}) = -\phi_1^2 -\phi_2^2 = \phi_3^2\).
    By the uniqueness statement of \Cref{lem:even_order}, we have \(\phi_3 = 2pq\) or \(\phi_3 = -2pq\).  In the latter case, we introduce a sign flip in \(p\).  Let \(\tilde p\coloneqq -p\).  Then \eqref{eq:Phi1Phi2ByPQ} remains intact \(\phi_1 = \tilde p^2 - q^2\), \(\phi_2 = -\ii(\tilde p^2 + q^2)\); while now \(\phi_3 = 2\tilde pq\).
\end{proof}

The following converse statement of \Cref{thm:ExistenceOfPythagoreanTriple} holds as well.  

If \(p,q\) are any holomorphic functions on \(\DD\), then \((\phi_1,\phi_2,\phi_3)\) defined by \eqref{eq:Pythagorean} is null holomorphic, and the surface \(\br\) given by \eqref{eq:HarmonicSurfaceFromNullVelocity} is a generalized minimal surface.  Moreover, \((\phi_1,\phi_2,\phi_3)\) vanishes at \(z\) only when \(p,q\) has a common zero at \(z\).  Thus, generically when the zeros of \(p,q\) do not coincide, the surface \(\br\) is a regular minimal surface.

In the non-generic case where the zeros of \(p,q\) do coincide, the zeros of \(\phi_1,\phi_2,\phi_3\) have an even order.  By \Cref{thm:RemovableBranch}, the generalized minimal surface \(\br\) has branched points that are all removable under perturbation.

\begin{corollary}\label{cor:PQToMinimal}
    For any holomorphic functions \(p(\zeta)\) and \(q(\zeta)\) on \(\DD\) and constant \(\bc\in\RR^3\), the surface
    \begin{align}\label{eq:RWrittenByPQ}
    \textstyle
        \br(z)\coloneqq \Re \int_0^z
        \Big(p^2-q^2,-\ii(p^2 + q^2),2pq\Big)\, d\zeta + \bc
    \end{align}
    is either a regular minimal surface (if \(p,q\) do not have common zero), or a generalized minimal surface whose branch points are all removable.
\end{corollary}

\subsection{Neural Representation of Minimal Surfaces}
\label{Sec:Neural_Architecture}

Building on \Cref{cor:PQToMinimal}, we translate Euclid's spinor representation of minimal surface into a trainable neural architecture. The overall pipeline is illustrated in \Cref{fig:NN_diagram}.

The two holomorphic functions $p$ and $q$ are represented by complex-valued neural networks.  A complex-valued neural network is a Multilayer Perceptrons (MLP) where each of the input, output, intermediate state variables and parameters are complex numbers.
It represent a complex-valued function of one complex variable when both the input and output is a single node.  The complex-valued function is guaranteed to be holomorphic if the activation functions are holomorphic.  We use the exponential function as the activation function.

Given holomorphic neural functions $p$ and $q$, we obtain $\bphi(z) = (\phi_1(z),\phi_2(z),\phi_3(z))$ using \eqref{eq:Pythagorean}.
Note that by \Cref{thm:ExistenceOfPythagoreanTriple} and the Universal Approximation Theorem for holomorphic functions \cite{calafa2024physics}, one can approximate any simply-connected minimal surface on any compact subset of $\mathbb{D}$ using this neural representation.
The minimal surface \(\br\colon\DD\to\RR^3\) is then evaluated as the real part of the integration of \((\phi_1,\phi_2,\phi_3)\) by \eqref{eq:HarmonicSurfaceFromNullVelocity}. 
In practice, this final integration is approximated by a numerical quadrature.  Specifically, we employ the trapezoidal rule.  Taking some number \(N\in\NN\) of partitions (we use \(N=100\) for all our examples), we evaluate the surface as
\begin{align}
\textstyle
    \br(z) = \Re\left[\sum_{k=1}^N\left( {\bphi((k-1){z\over N})+\bphi(k{z\over N})}\right){z\over 2N}\right],\quad z\in\DD.
\end{align}
The error introduced by this numerical quadrature is \(O(1/N^2)\) (as shown in \Cref{fig:compare_int_N_analytic}), while higher order approximations are also widely available, \eg\@ Simpson's rule yields an error of \(O(1/N^4)\).

Every instance of our neural minimal surface (even without any training with all parameters of the MLPs randomized) is a \(O(1/N^2)\) distant away from an exact minimal surface.  By \Cref{cor:PQToMinimal}, this minimal surface is generically regular.


\section{Approach to the Plateau Problem}
Let us return to the Plateau problem (\Cref{prob:plt_prob}).  Here, we propose an optimization-based approach.  The unknown surface \(\br\) is represented by a neural network as described in \Cref{Sec:Neural_Architecture}, which always yields a minimal surface.  The remaining task is to ensure that the boundary image \(\br|_{\partial\DD}\) matches a given space curve \(\Gamma\colon\SS^1\to\RR^3\).

A tempting loss function for this fidelity problem is \(\|\br|_{\partial\DD} - \Gamma\|^2\).  
However, requiring the boundary map \(\br|_{\partial\DD}\) to be equal to the parameterized curve \(\Gamma\) makes the problem over-constrained.  Note that the original problem only requires the two curves to fit as unparameterized curves. 
The appropriate formulation is to require \(\br|_{\partial\DD}\) to coincide with \(\Gamma\circ\tilde\Phi\) for some reparameterization \(\tilde\Phi\in\Homeo^+(\SS^1)\), where \(\Homeo^+(\SS^1)\) denotes the space of all orientation-preserving homeomorphisms from \(\SS^1\) to itself (\Cref{fig:reparametrization}).






\begin{problem} \label{optim}
    Given an arbitrary rectifiable Jordan curve $\Gamma:\SS^1 \to \RR^3$, find two holomorphic functions \(p,q\) on \(\DD\) and \(\tilde\Phi\in\Homeo^+(\SS^1)\) that solve
\begin{equation}
\textstyle
{\rm minimize}_{(p,q),\tilde\Phi}
\int_{\mathbb{S}^1} \left| \br|_{\mathbb{S}^1} - \Gamma \circ \tilde{\Phi} \right|^2d\mu,
\label{eq:energy}
\end{equation}
where \(\br\) is defined in terms of \(p,q\) by \eqref{eq:RWrittenByPQ}, and \(\mu\) is any measure on \(\SS^1\).
\end{problem}

\subsection{Reparameterizations on \(\SS^1\)}
\label{Sec:homeo}

To incorporate the homeomorphism \(\tilde\Phi\), we adopt another neural network to represent the homeomorphism.  
First, identify \(\SS^1\) with \(\RR/\ZZ\).
Then every orientation-preserving homeomorphism on \(\SS^1\) is equivalent to an increasing function \(\Psi\) from \([0,1]\) onto \([b,b+1]\) for some \(b\in\RR\).
Every such function can be written as 
\begin{align}
    \label{eq:homeo_NN}
    \textstyle
    \Psi(t)\coloneqq({\int_{0}^{t} \rho(s) \,ds})/({\int_{0}^{1} \rho(s) \,ds}) + b
\end{align}
for some positive-valued function \(\rho\colon [0,1]\to\RR_{>0}\).  We let \(\rho\) be a real-valued single-input-single-output MLP with ReLU activations for hidden layers, and an exponential function as the last activation layer (to ensure positive values).  We also let \(b\in\RR\) be a learnable parameter.

We show that every orientation-preserving homeomorphism between \([0,1]\) and \([b,b+1]\) can be approximated by some \(\Psi\) defined by \eqref{eq:homeo_NN}.

\begin{theorem} \label{0-1 to 0-1 homeo_NN}
    Let $\Phi\colon [0,1] \rightarrow[0,1]$ be an increasing homeomorphism. For any $\epsilon >0$, there exists $\rho\colon [0,1] \rightarrow \mathbb{R}_{>0}$, an MLP with ReLU activations in hidden layers and exponential map as output activation, such that $\|\Psi-\Phi\|_{\infty}<\epsilon$ where $\Psi(t)=({\int_{0}^{t} \rho(s) \,ds})/({\int_{0}^{1} \rho(s) \,ds})$.
\end{theorem}
\begin{proof}
    See Appendix~\ref{app:proofofhomeo}.
\end{proof}
By adding a learnable real parameter \(b\), \Cref{0-1 to 0-1 homeo_NN} implies the following approximation theorems for arbitrary orientation-preserving homeomorphisms on \(\SS^1\).
\begin{corollary}\label{cor:HomeoNNShift}
    For each increasing function \(\Phi\colon[0,1]\to\RR\) with \(\Phi(1)-\Phi(0)=1\), and \(\epsilon>0\), there exists an MLP \(\rho\colon[0,1]\to\RR_{>0}\) as described in \Cref{0-1 to 0-1 homeo_NN} and \(b\in\RR\) such that \(\|\Psi - \Phi\|_{\infty}<\epsilon\) where \(\Psi(t)=({\int_{0}^{t} \rho(s) \,ds})/({\int_{0}^{1} \rho(s) \,ds})+b\).
\end{corollary}
\begin{corollary}\label{cor:HomeoNNS1}
    For any \(\tilde\Phi\in\Homeo^+(\SS^1)\) and \(\epsilon>0\), there exists an MLP \(\rho\) and \(b\in\RR\) representing \(\Psi\) as in  \Cref{cor:HomeoNNShift} so that \(\|e^{\ii 2\pi \Psi} - \tilde\Phi\|_{\infty}<\epsilon\).
\end{corollary}

\begin{figure}
    \centering
    \setlength{\unitlength}{1pt}
    \begin{picture}(0.85\columnwidth,135)
        \put(0,-10){
        \includegraphics[width=0.8\columnwidth]{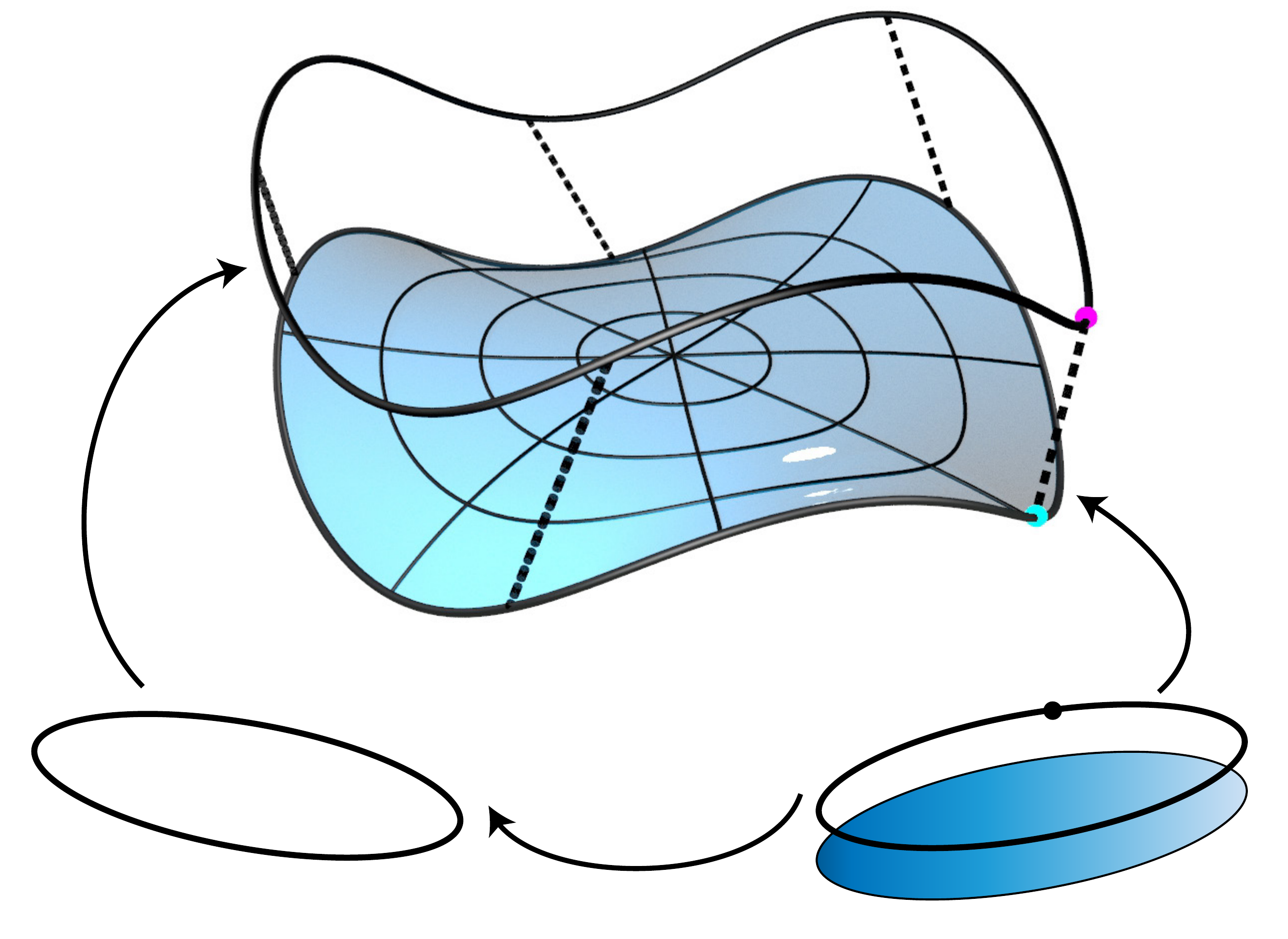}
            }
        \put(97,5){$\tilde{\Phi}$}
        \put(190,0){$\mathbb{D}$}
        \put(190,25){$\partial\mathbb{D}$}
        \put(183,47){$\br|_{\partial\mathbb{D}}$}
        \put(160,31){$p$}
        \put(20,60){$\Gamma$}
        \put(145,45){$\br|_{\partial\mathbb{D}}(p)$}
        \put(170,85){$\Gamma\circ\tilde{\Phi}(p)$}
        \put(30,-3.5){$\mathbb{S}^1$}
        
    \end{picture}
    \caption{Given the curve \(\Gamma\colon\SS^1\to\RR^3\), we seek a minimal surface \(\br\colon\DD\to\RR^3\) together with a homeomorphism \(\tilde\Phi\) between \(\partial\DD\) and \(\SS^1\) that minimizes the distance between the boundary map \(\br|_{\partial\DD}\) and \(\Gamma\circ\tilde\Phi\). 
    }
    \label{fig:reparametrization}
\end{figure}

\subsection{Numerical Algorithm}
\Cref{Sec:homeo} shows that homeomorphisms \(\SS^1\) can be represented by an MLP \(\rho\).
This allows the problem \eqref{eq:energy} to be formulated as an optimization for the parameters in the neural networks. Let \(w_{p}, w_{q}, w_{\rho}\) denote the parameters for the network representing functions \(p, q,\) and \(\rho\), respectively.
Let \(\cL( \br|_{\partial\DD},\Gamma\circ\tilde\Phi)\) denote the loss function that measures the fidelity of the two functions \(\br|_{\partial\DD},\Gamma\circ\tilde\Phi\).  For example \(\cL\) can be given by the integral in \eqref{eq:energy}, or some higher order norm detailed in \Cref{sec:min_surface_training}. In either case, \(\cL( \br|_{\partial\DD},\Gamma\circ\tilde\Phi)\) is a function of the neural network parameters \(w_p, w_q, w_\rho\).

To minimize \(\cL(\br|_{\partial\DD},\Gamma\circ\tilde\Phi)\), we alternate between (i) updating the surface network \(\br\) parameterized by \({w_p}\) and \({w_q}\), and (ii) the homeomorphism module \(\tilde\Phi\) with parameters \(w_\rho\) (\Cref{overall_alg}).

However, note that if the given curve \(\Gamma\) is not parameterized differentiably, and \(\tilde\Phi\) is represented by the MLP \(\rho\), then the loss function \(\cL(\br|_{\partial\DD},\Gamma\circ\tilde\Phi)\) may not be differentiable with respect to \(w_\rho\) during step (ii).

To overcome this, in step (ii), we minimize \(\tilde\cL\coloneqq \cL(\br|_{\partial\DD}\circ\tilde\Phi^{-1},\Gamma)\) instead.  In practice, it is the homeomorphism \(\tilde\Phi^{-1}\) that is parameterized by \(\rho\) using the method of \Cref{Sec:homeo}.  This ensures that in step (ii), \(\tilde\cL\) is differentiable with respect to \(w_\rho\) (note that \(\br|_{\partial\DD}\) is differentiable).  In step (i), the inverse \(\tilde\Phi\) of \(\tilde\Phi^{-1}\) is evaluated via a bisection method.  Note that in step (i), only the parameters for \(\br\) in \(\cL(\br|_{\partial\DD},\Gamma\circ\tilde\Phi)\) needs to be differentiated, and \(\Gamma\circ\tilde\Phi\) can remain non-differentiable.

If the given \(\Gamma\) is differentiable, then we simply set \(\tilde\cL\) as the original objective \(\tilde\cL = \cL(\br|_{\partial\DD},\Gamma\circ\tilde\Phi)\) and parameterize \(\tilde\Phi\) by the MLP \(\rho\).


\begin{algorithm}
\caption{Neural Optimization of Minimal Surfaces}
\begin{algorithmic}[1]
    \Require Boundary curve $\Gamma\colon\SS^1\to\RR^3$
    \Ensure Minimal surface $\br\colon\DD\to\RR^3$
    
    \State Initialize neural networks $p_{w_p},q_{w_q}$ (holomorphic functions) and $\rho_{w_{\rho}}$ (density function for boundary homeomorphism \(\tilde\Phi^{-1}\))
    
    \While{not converged}
        \State \textbf{Surface Update}
        \State \hspace{\algorithmicindent} Sample points $\{z_i\}$ on the unit circle $\mathbb{S}^1$;
        \State \hspace{\algorithmicindent} Evaluate $\Gamma \circ \tilde{\Phi}(z_i)$ and $\br(z_i)$;
        \State \hspace{\algorithmicindent} Evaluate $\cL$ and update $w_p,w_q$ using RMSprop;
        \State \textbf{Boundary Homeomorphism Update}
        \State \hspace{\algorithmicindent} Sample points $\{z_i\}$ on the unit circle $\mathbb{S}^1$;
        \State \hspace{\algorithmicindent}
        Evaluate \(\br(\tilde\Phi^{-1}(z_i))\) and \(\Gamma(z_i)\);
        \State \hspace{\algorithmicindent} 
        {Evaluate  $\tilde\cL$ and update $w_\rho$ using RMSprop};
    \EndWhile
    \State \Return minimal surface $\br$
\end{algorithmic}
\label{overall_alg}
\end{algorithm}

\subsubsection{The Loss Function}
\label{sec:min_surface_training}


The loss function \(\cL(\bx,\by)\) for functions \(\bx,\by\colon\SS^1\to\RR^3\) is designed to be the fidelity energy in \eqref{eq:energy} potentially with the addition of a higher order term.  Specifically, let \(\cL_{\rm distance}(\bx,\by) \coloneqq \int_0^1|\bx(e^{2\pi\ii t}) - \by(e^{2\pi\ii t})|^2\, dt\).  For the higher order term, let \(\bv(t)\coloneqq {d\over dt}\bx(e^{2\pi\ii t})\) and \(\bw(t)\coloneqq {d\over dt}\by(e^{2\pi\ii t})\), and define 
\(\cL_{\rm tangent}(\bx,\by) = \int_0^1 (1-{\bv\over |\bv|}\cdot{\bw\over|\bw|} )\, dt\) which measures the fidelity of the curves' tangents.
We let \(\cL(\bx,\by)\coloneqq\cL_{\rm distance}(\bx,\by) + \lambda\cL_{\rm tangent}(\bx,\by)\), where \(\lambda \geq 0\) is a nonnegative weight. \Cref{fig:vel_comparison} shows that a nonzero $\lambda$ improves the optimization.

The integrals over \(t\in[0,1]\) in \(\cL(\bx,\by)\) are evaluated via Monte Carlo integration with random sampling on \([0,1]\).  We optimize \(\cL\) using the RMSprop optimizer.

\section{Results}
\label{sec:results}

We demonstrate how our method approaches the Plateau problem through a series of numerical experiments.
All experiments share the same hyperparameters. The networks $p$ and $q$ are modeled as MLPs with $10$ layers and width $20$ (input/output dimension 1). The reparameterization network $\rho$
uses an MLP with architecture $[1,64,64,32,1]$. We use a learning rate of $10^{-5}$, a batch size of $128$, $1024$ random points for training generated at each epoch, a coefficient $\lambda=0.5$, and a number of trapezoidal integration steps $N=100$. (The code is included in the supplementary material.) 

\subsection{Comparison with the Weierstrass--Enneper Para\-me\-terization} 
We compare our method with computing minimal surfaces using the classical Weierstrass-Enneper Parameterization \eqref{eq:WE}, where $f$ and $g$ are parameterized by two CVNNs with the same architecture as $p$ and $q$ in our method.

\begin{figure}
    \centering
    
    \begin{subfigure}[t]{0.48\linewidth}
        \centering
        \includegraphics[width=\linewidth]{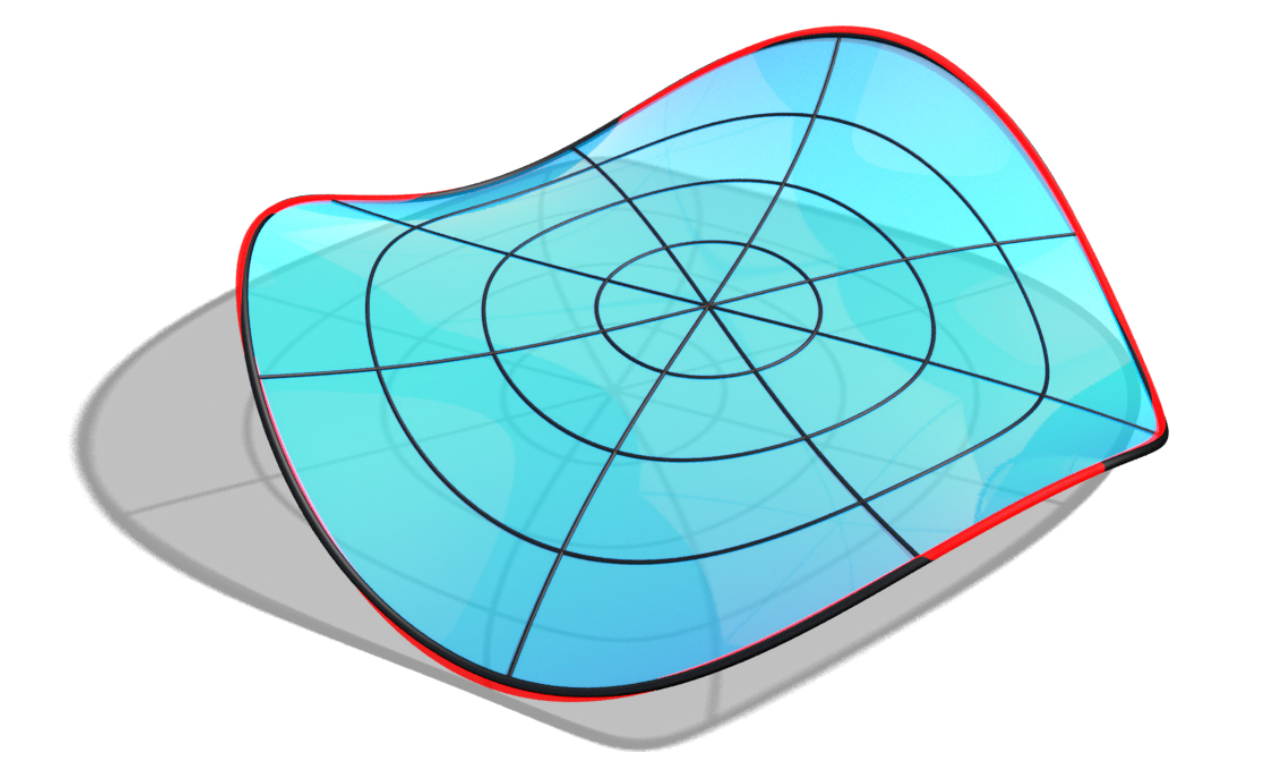}
        \caption{\textbf{Our method}}
        \label{fig:compare_ours}
    \end{subfigure}
    \hfill
    \begin{subfigure}[t]{0.48\linewidth}
        \centering
        \includegraphics[width=\linewidth]{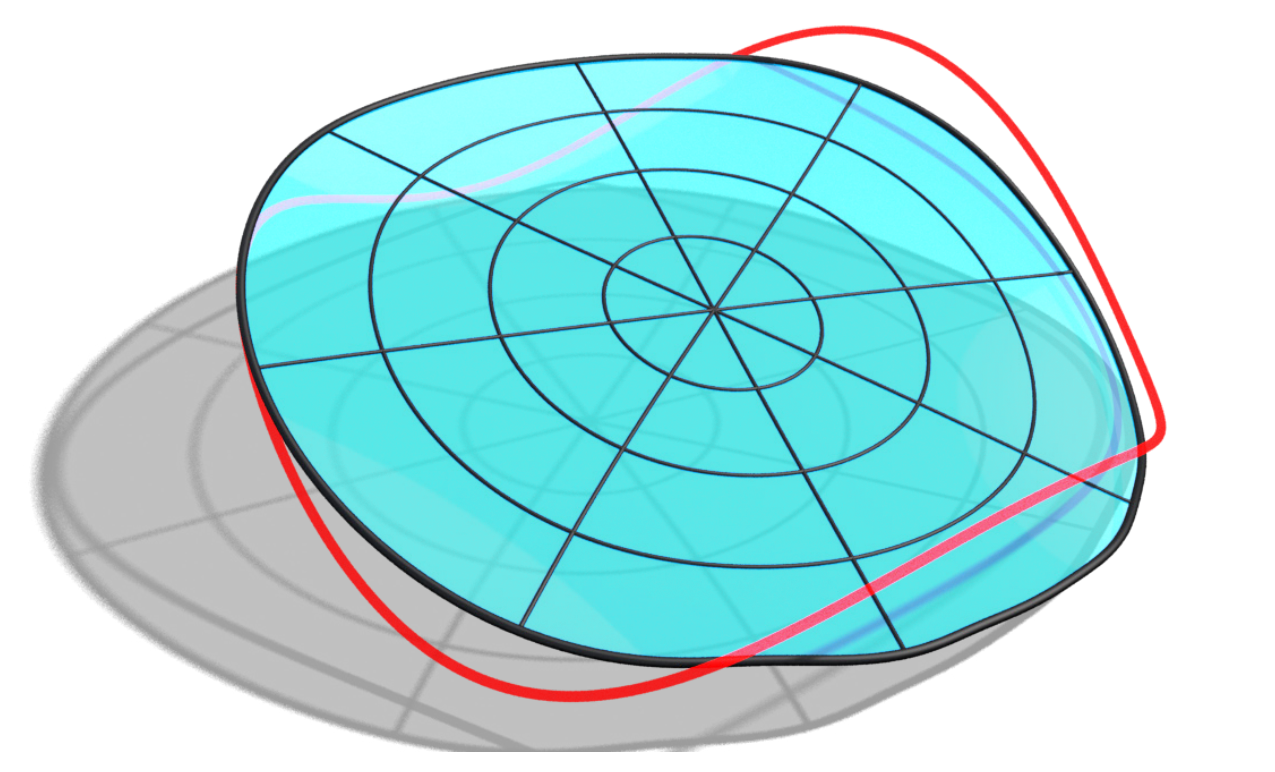}
        \caption{Weierstrass--Enneper}
        \label{fig:compare_WE}
    \end{subfigure}
    
    \caption{Comparison of Plateau problem optimization using (a) our method and (b) a neural Weierstrass--Enneper parameterization. The prescribed boundary curve, shown in red, is a sinusoidally perturbed circle.}
    \label{fig:comparison}
\end{figure}

\Cref{fig:comparison} shows the comparison of the two methods with the boundary curve given by a sinusoidally perturbed circle:
\begin{equation} \label{twisted_circ_eq}
\Gamma : [0, 2\pi] \to \mathbb{R}^3, \quad 
t \mapsto \bigl( R \cos t,\; R \sin t,\; \alpha \sin(k t) \bigr)
\end{equation}
where $R=1$ is the radius, $k=3.0$ controls the oscillation frequency, and $\alpha=0.2$ determines the vertical amplitude. 
The method based on the Weierstrass--Enneper parameterization fails to find the minimal surface bounded by the curve, since its Gauss map contains directions that induce pole singularities in the function \(g\).  Our method based on Euclid's spinor representation is able to produce a minimal surface whose boundary is close to the given curve.

\Cref{fig:comparison2} shows a similar comparison with \(\Gamma\) being a tennis-seam curve, given by the intersection of the unit sphere with a cylinder.

All experiments are trained for 800 epochs.

\begin{figure}
    \centering
    
    \begin{subfigure}[t]{0.48\linewidth}
        \centering
        \includegraphics[width=0.8\linewidth, trim=100 0 100 0,
            clip]{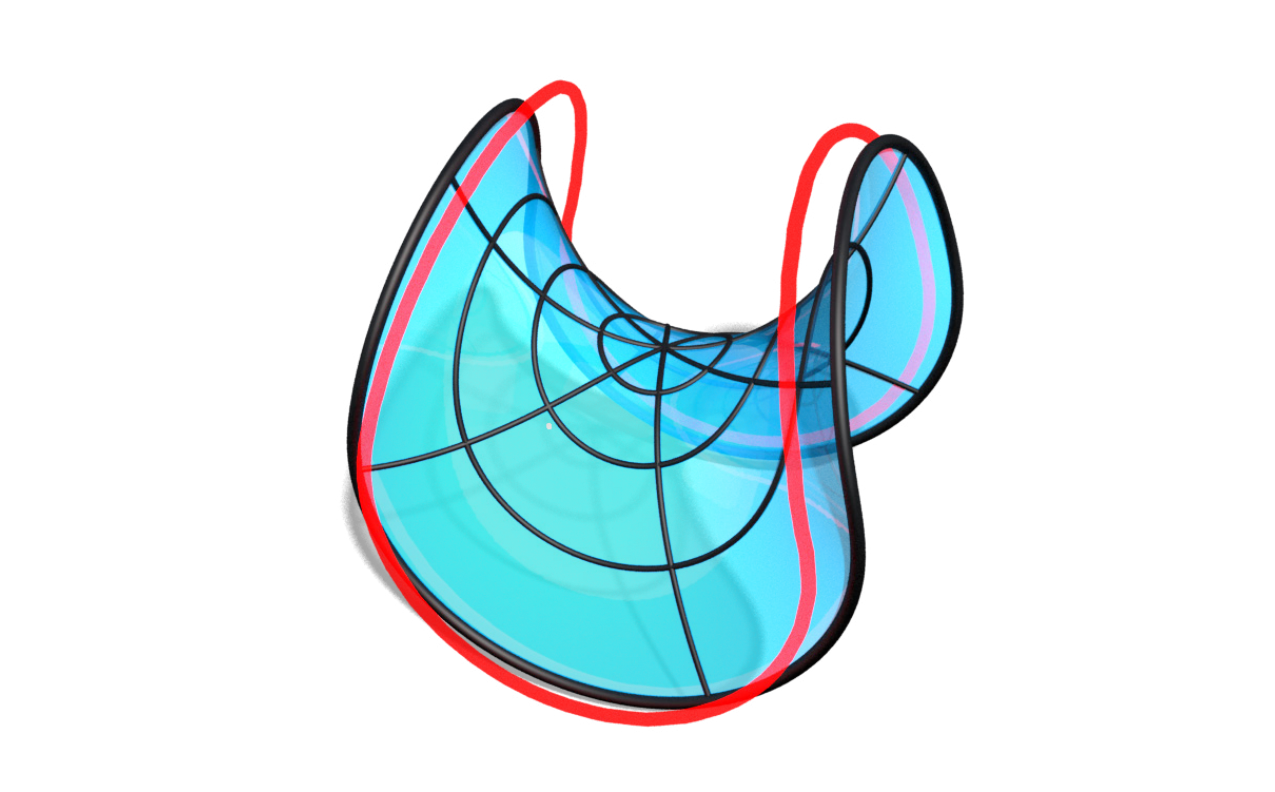}
        \caption{\textbf{Our method}}
        \label{fig:compare_ours}
    \end{subfigure}
    \hfill
    \begin{subfigure}[t]{0.48\linewidth}
        \centering
        \includegraphics[width=0.8\linewidth, trim=100 0 100 0,
            clip]{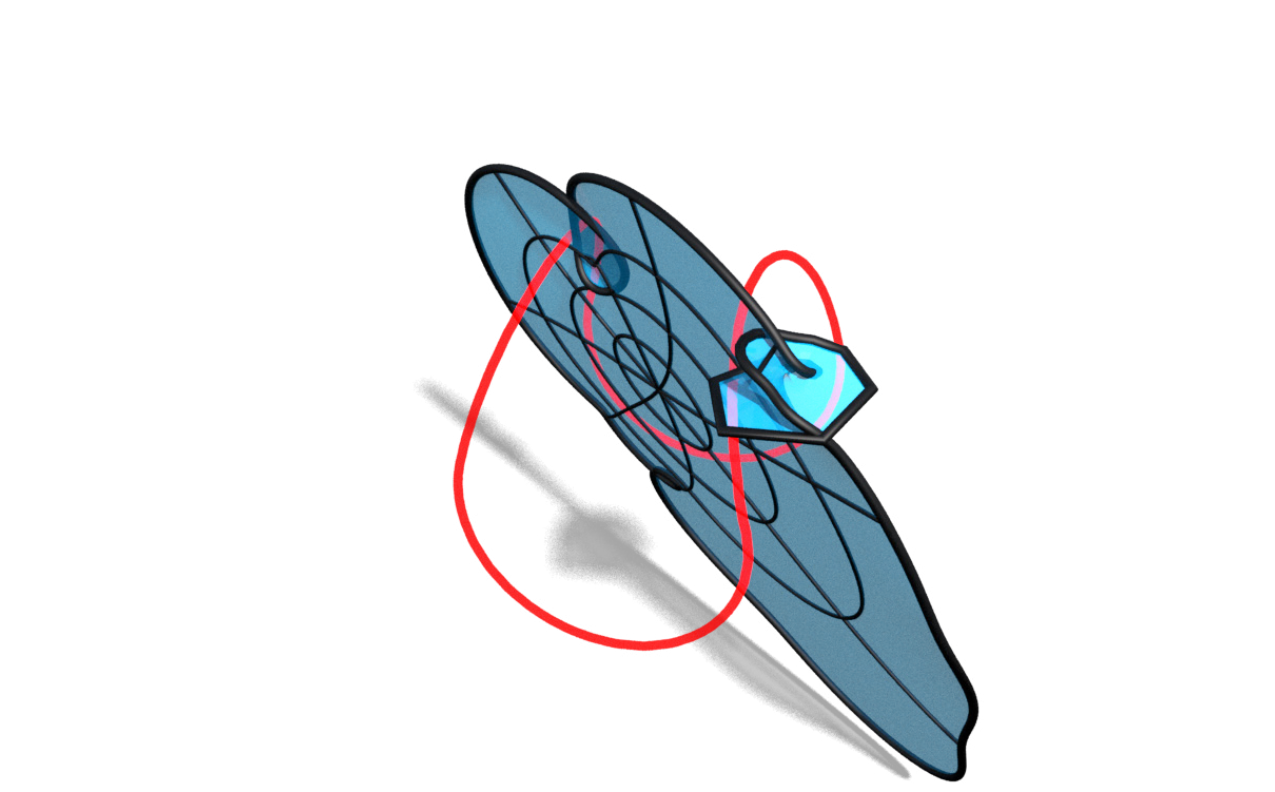}
        \caption{Weierstrass-Enneper}
        \label{fig:compare_WE}
    \end{subfigure}
    
    \caption{
    Comparison of Plateau problem optimization using (a) our method and (b) a neural Weierstrass--Enneper parameterization. The prescribed boundary curve is a tennis-seam curve.}
    \label{fig:comparison2}
\end{figure}

\subsection{Validation with Prescribed Minimal Surfaces}

For a fixed boundary \(\Gamma\), the solution to the Plateau problem is generically locally unique (up to M\"obius reparameterization on \(\DD\)).%
\footnote{The second variation of the area functional is given by the operator \(\Delta + |A|^2\) where \(\Delta\) is the negative definite Laplacian, which is generically non-degenerate.}
Thus, by taking \(\Gamma\) as the boundary of a prescribing a minimal surface (through prescribing functions $p$ and $q$), we should able to recover the the same minimal surface (\ie\@ the $p,q$ data up to M\"obius reparameterizations). 



\Cref{fig:comparison_with_prescribed_surface} shows this experiment on two different prescribed minimal surfaces. The surface of \Cref{fig:comparison_with_prescribed_surface}~(left) is obtained with neural $p$ and $q$ with randomized parameters. To make the \(p,q\) data more generic outside of the neural represented space, the $p,q$ data for \Cref{fig:comparison_with_prescribed_surface}~(right) is obtained by extracting the 10th order Taylor expansions of the $p$ and $q$ from
that of \Cref{fig:comparison_with_prescribed_surface}~(left), and randomly perturbing the
coefficients.  In both cases, our optimization method are able to closely recover the prescribed surface with matching \(p,q\) data and Weierstrass--Enneper data \(f,g\) of the surface up to M\"obius transformations (nearly just rotations in both cases) (Figures~\ref{fig:pq_random} and \ref{fig:pq_perturbed}).  

\begin{figure}
    \begin{subfigure}{0.49\linewidth}
        \centering
        \includegraphics[
            width=\linewidth,
            trim=0 0 50 0,
            clip
        ]{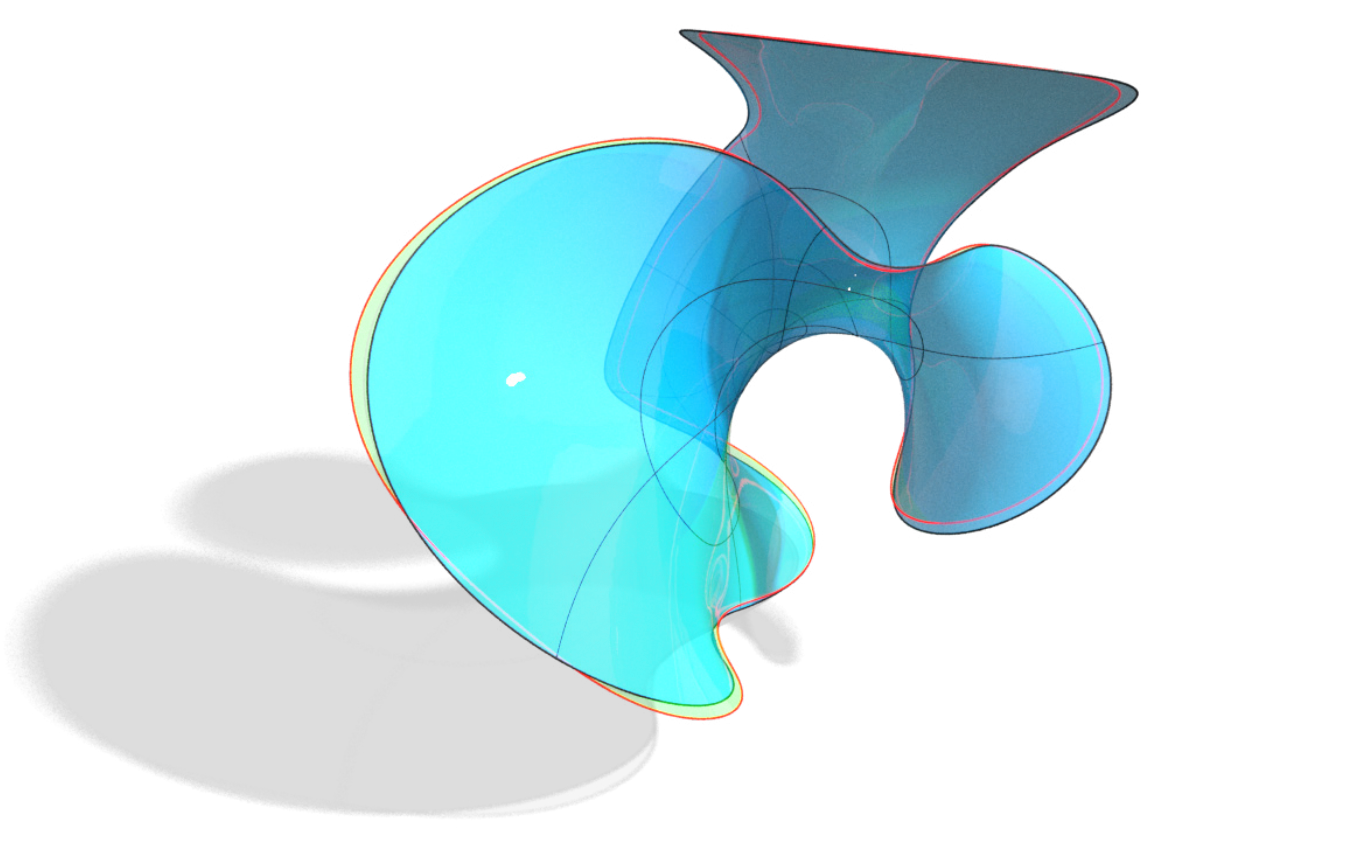}
        
        \label{fig:random_surface}
    \end{subfigure}
    \hfill
    \begin{subfigure}{0.49\linewidth}
        \centering
        \includegraphics[
            width=\linewidth,
            trim=0 0 30 0,
            clip
        ]{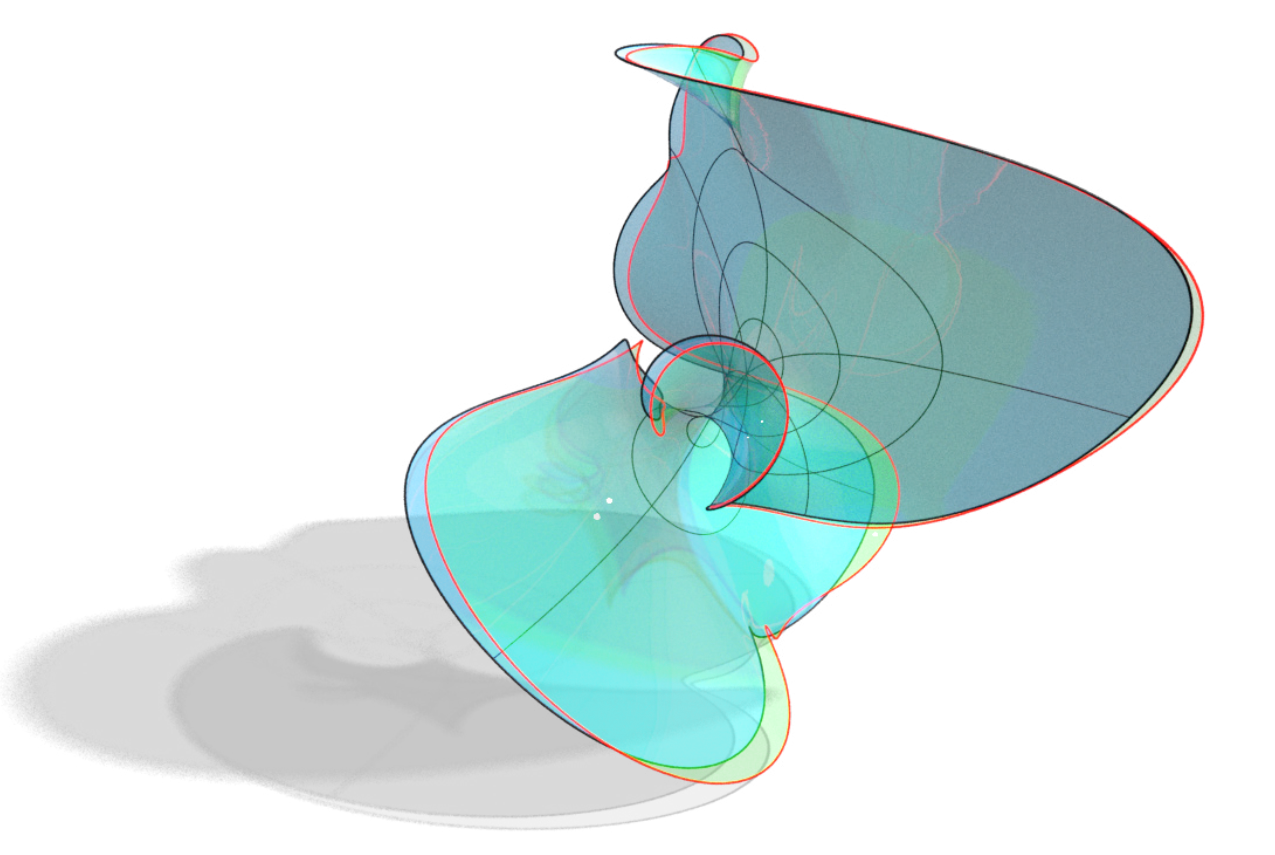}
        
        \label{fig:perturbed_surface}
    \end{subfigure}
    \caption{Two comparisons of the trained surface (blue interior, black boundary) and the prescribed ground truth (green interior, red boundary). The prescribed ground truth is synthesized by: a neural \(p\), \(q\) with randomized parameters (left); perturbing the Taylor expansion of the \(p,q\) data of the left example (right). The models were trained for 2000 epochs (left) and 4000 epochs (right).}
    \label{fig:comparison_with_prescribed_surface}
\end{figure}

\begin{figure}
    \centering
    \setlength{\unitlength}{1pt}
    \begin{picture}(\columnwidth,120)
    \put(0,0){
            \includegraphics[width=\columnwidth]{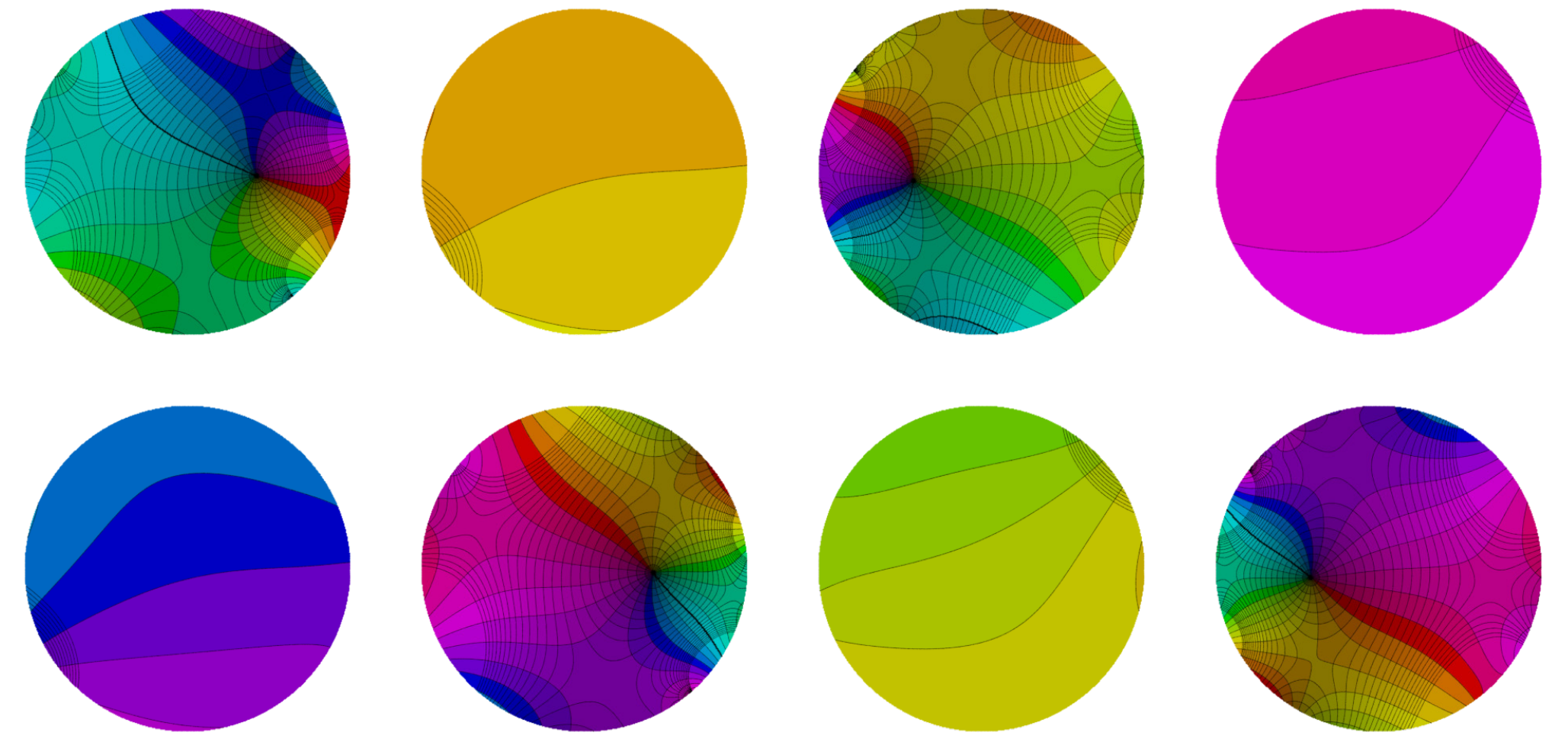}
            }
    \put(57,0){\small\sffamily(i)}
    \put(180,0){\small\sffamily(ii)}
    \put(30,55){$f$}
    \put(30,117){$p$}
    \put(91,117){$q$}
    \put(91,55){$g$}
    \put(153,55){$f$}
    \put(213,55){$g$}
    \put(153,117){$p$}
    \put(213,117){$q$}

    \end{picture}
    \caption{Corresponding functions p, q, f, and g for (i) the trained surface, and (ii) the ground truth of the surface in \Cref{fig:comparison_with_prescribed_surface}~(left).}
    \label{fig:pq_random}
\end{figure}

\begin{figure}[h]
    \centering
    \setlength{\unitlength}{1pt}
      \begin{picture}(\columnwidth,120)
        \put(0,0){
                \includegraphics[width=\columnwidth]{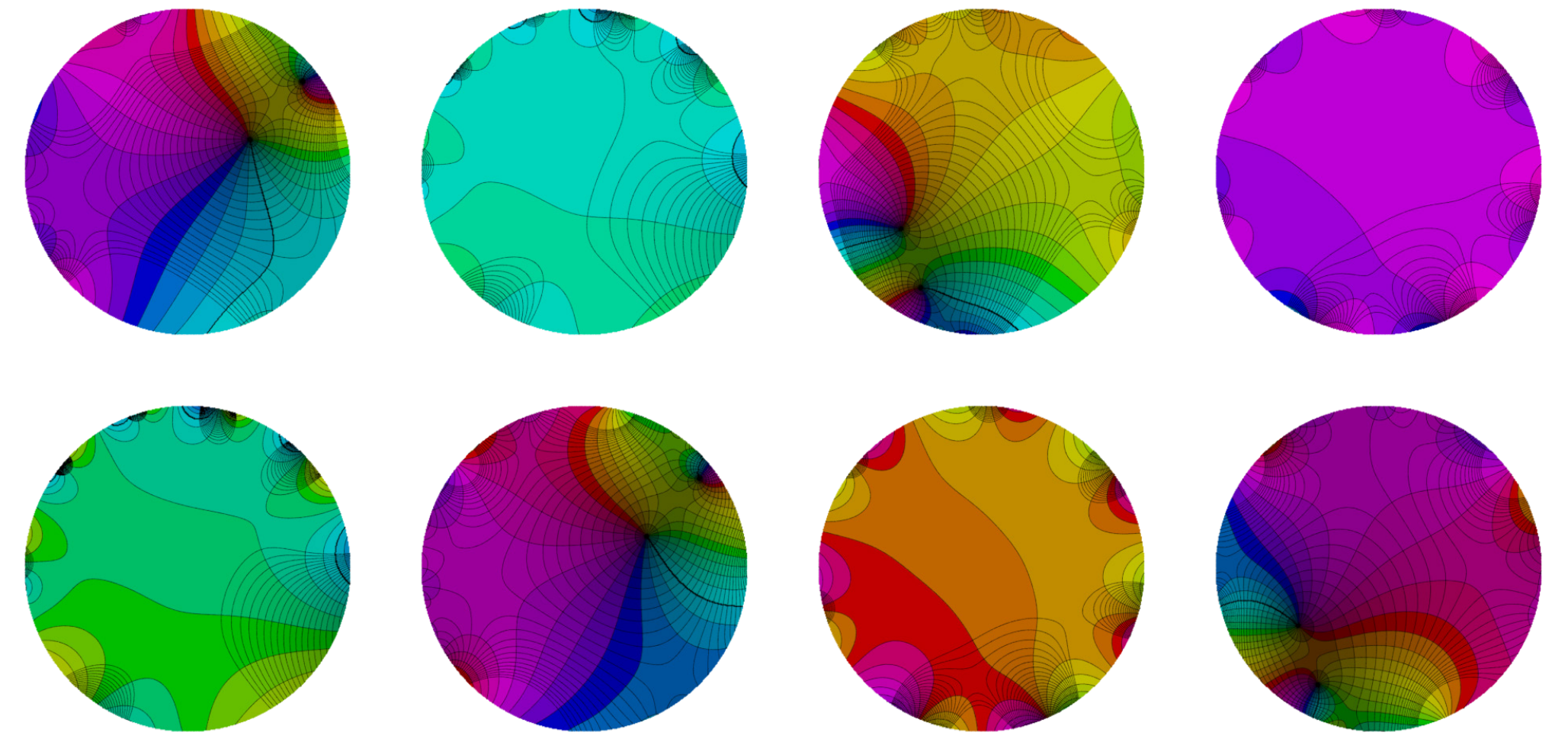}
                }
        \put(57,0){\small\sffamily(i)}
        \put(180,0){\small\sffamily(ii)}
        \put(30,55){$f$}
        \put(30,117){$p$}
        \put(91,117){$q$}
        \put(91,55){$g$}
        \put(153,55){$f$}
        \put(213,55){$g$}
        \put(153,117){$p$}
        \put(213,117){$q$}

      \end{picture}

    \caption{Corresponding functions p, q, f, and g for (i) the trained surface, and (ii) the ground truth of the surface in \Cref{fig:comparison_with_prescribed_surface}~(right).}
    
    \label{fig:pq_perturbed}
\end{figure}

\subsection{Surface Minimality Evaluation via Mean Curvature}

For each of our neural minimal surface, we compute its mean curvature through auto differentiation.    \Cref{fig:mean_curvature_heatmaps} show nearly vanishing mean curvature across the surface, indicating that these surfaces approximate minimal surfaces.

Note that the only error in the vanishing mean curvature condition arises in the trapezoidal quadrature, not from the training process or evaluation of the loss function. \Cref{fig:Meancurvature_vs_intsteps} shows that \(\|H\|_{\max}\) converges to zero with \(O({1\over N^2})\) with increasing numbers $N$ of quadrature partitions.




\section{Conclusion}

We presented a neural representation for constructing minimal surfaces.  Unlike approaches that optimizes to approximate the minimal surface equations, our representation intrinsically guarantees the minimality of the generated surface. 
By optimizing for boundary consistency, we demonstrate its ability to solve the Plateau problem. Experimental results show that the surfaces closely match the prescribed boundaries, and the zero mean curvature condition is achieved with second-order accuracy. In addition to the surface itself, the method also produces the associated Weierstrass data \((\phi_1,\phi_2,\phi_3)\), which is the analytic fingerprint of the minimal surface.

Our current method nevertheless has several limitations. First, it does not guarantee exact boundary matching. Second, the formulation currently assumes the domain is the unit disk, which restricts the method to minimal surfaces of disk topology and prevents the direct generation of surfaces with arbitrary topology.

Despite these limitations, the results demonstrate that analytic representations can be effectively integrated with modern neural networks to produce geometrically faithful surface models. 
More broadly, this work suggests a promising direction for combining exact structures of the problem with neural representations, enabling learning-based methods that preserve intrinsic geometric properties by construction.


\bibliographystyle{ACM-Reference-Format}
\bibliography{Reference}  

\newpage

\newpage

\begin{figure}[t]
    \centering
    \includegraphics[width=\linewidth]{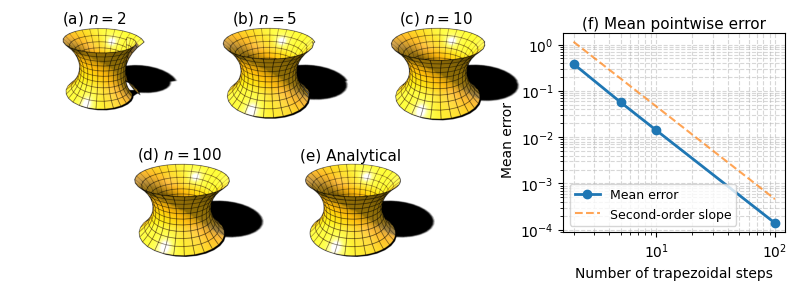}
   \caption{Comparison of trapezoidal integration with varying numbers of quadrature steps against the analytical solution of the catenoid. The parameter domain is $[-1,1]\times[0,2\pi]$. The error is measured as the pointwise distance between corresponding surface points sampled on a $120\times30$ grid.}
    \label{fig:compare_int_N_analytic}
\end{figure}

\begin{figure}
    \centering
    \begin{subfigure}{0.48\linewidth}
        \centering
        \includegraphics[width=\linewidth]{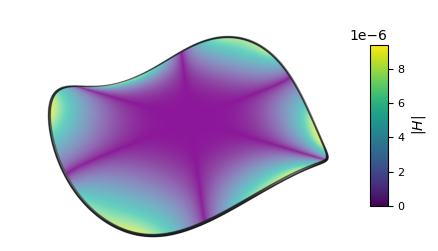}
        \label{fig:twsd_circ_mean_curvature}
    \end{subfigure}
    \hfill
    \begin{subfigure}{0.48\linewidth}
        \centering
        \includegraphics[width=\linewidth]{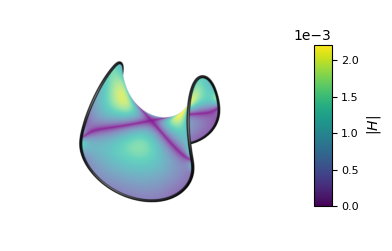}
        \label{fig:seam_curve_mean_curvature}
    \end{subfigure}

    \vspace{0.5em}

    \begin{subfigure}{0.48\linewidth}
        \centering
        \includegraphics[width=\linewidth]{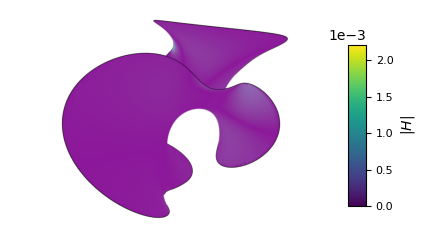}
        \label{fig:rand_surface_mean_curvature}
    \end{subfigure}
    \hfill
    \begin{subfigure}{0.48\linewidth}
        \centering
        \includegraphics[width=\linewidth]{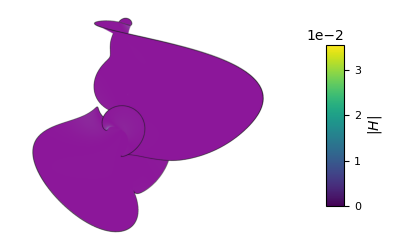}
        \label{fig:perturbed_surface_mean_curvature}
    \end{subfigure}

    \caption{Visualization of the absolute mean curvature $|H|$ on reconstructed surfaces for four different boundary conditions. The uniformly near-zero mean curvature values across the surfaces indicate that the reconstructed geometries closely approximate minimal surfaces.}
    \label{fig:mean_curvature_heatmaps}
\end{figure}

\begin{figure}
    \centering
    \setlength{\tabcolsep}{1pt}
    \renewcommand{\arraystretch}{1}
    \resizebox{0.85\linewidth}{!}{
    \begin{tabular}{@{}r c c c c c c@{}}
        & \scriptsize epoch 0 & \scriptsize epoch 2 & \scriptsize epoch 4
        & \scriptsize epoch 6 & \scriptsize epoch 8 & \scriptsize epoch 10 \\[-0.2em]

        \makebox[0pt][r]{\raisebox{16pt}{\scriptsize $\lambda=0$}} &
        \includegraphics[width=0.125\linewidth, trim=140 0 140 0, clip]{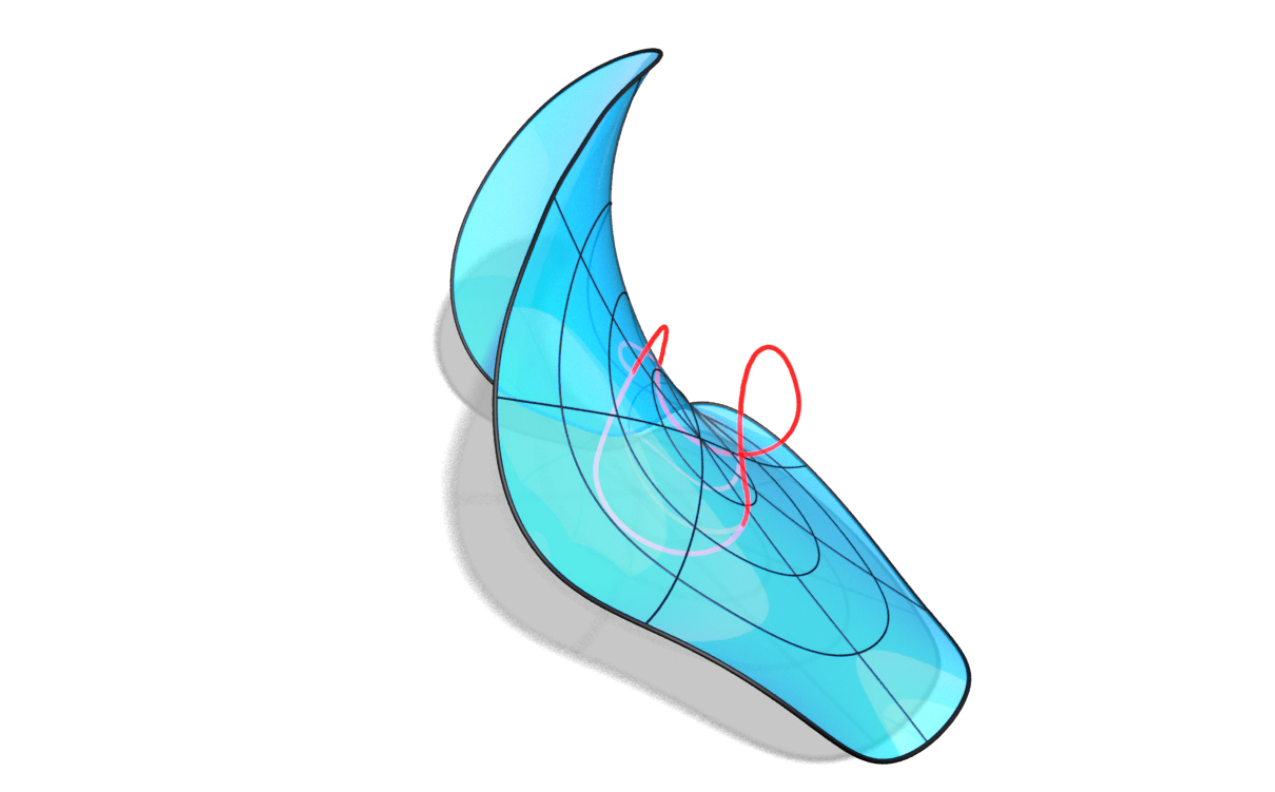} &
        \includegraphics[width=0.125\linewidth, trim=140 0 140 0, clip]{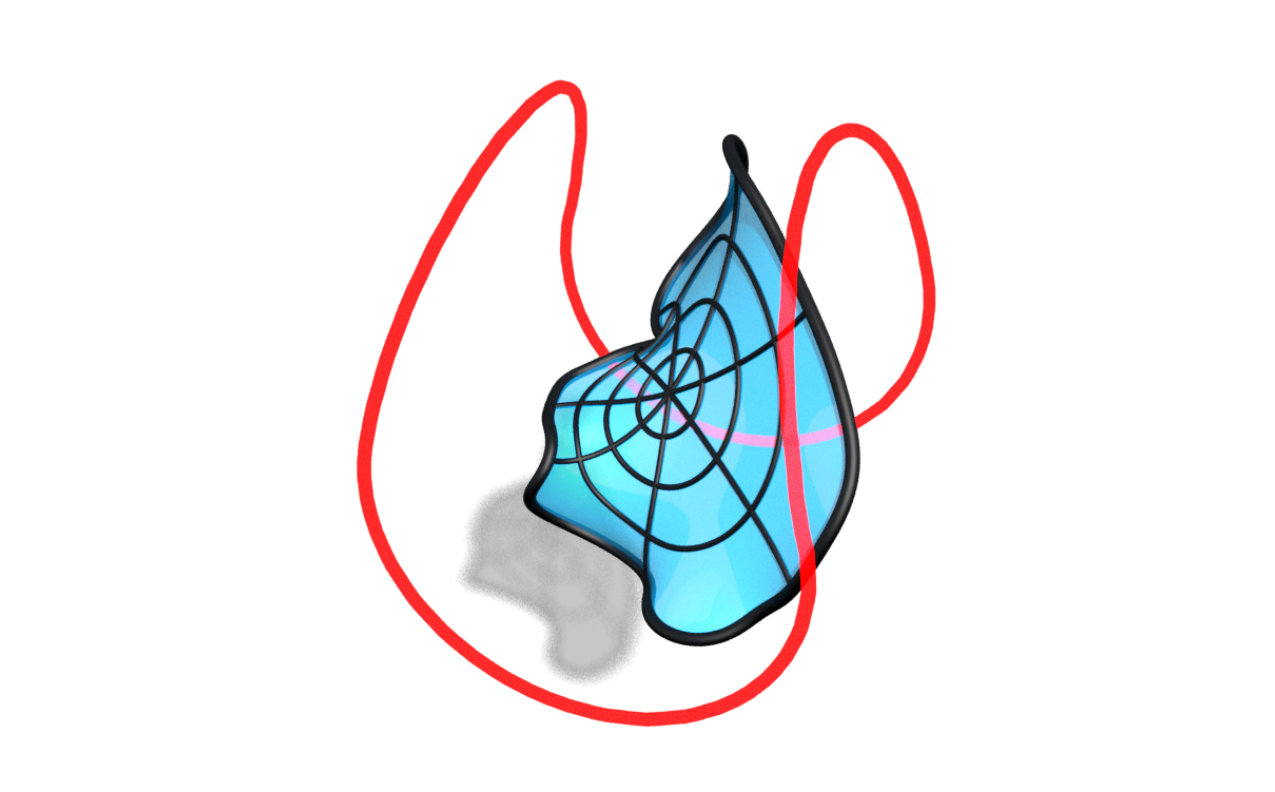} &
        \includegraphics[width=0.125\linewidth, trim=140 0 140 0, clip]{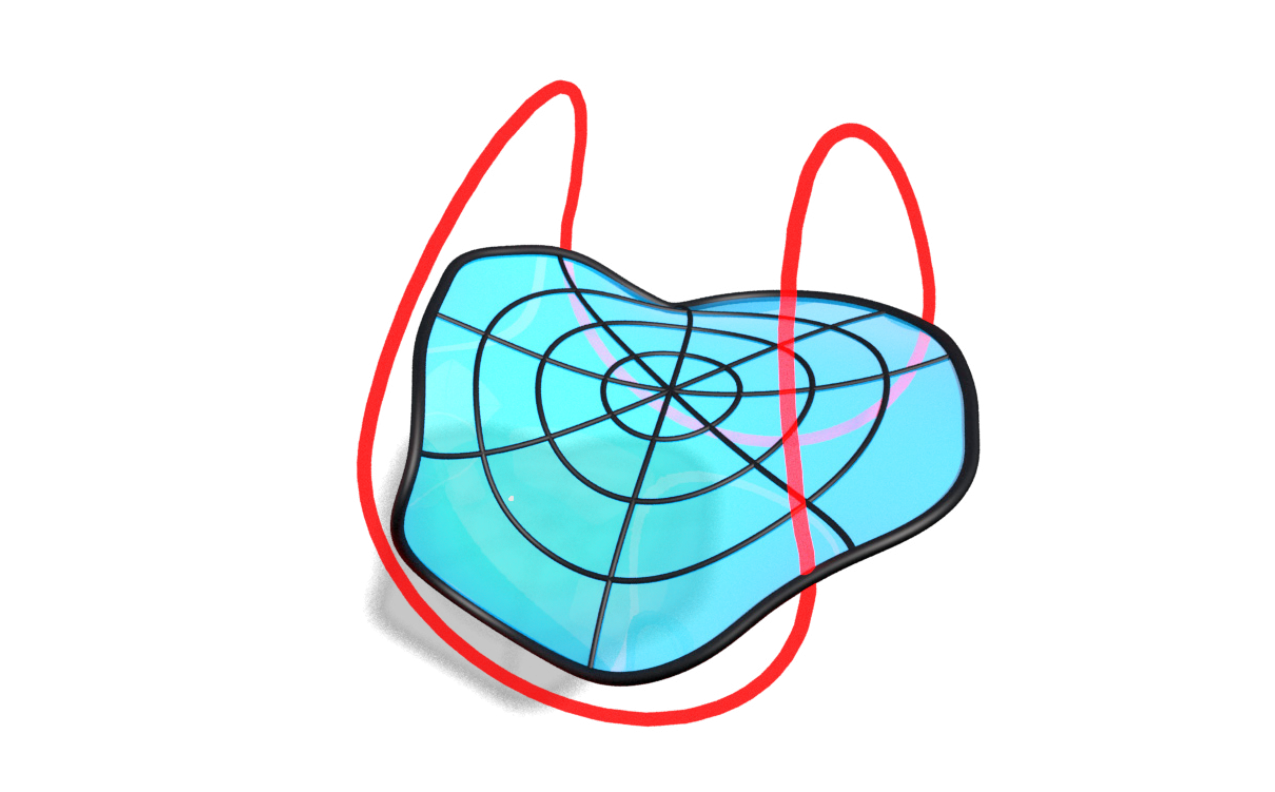} &
        \includegraphics[width=0.125\linewidth, trim=140 0 140 0, clip]{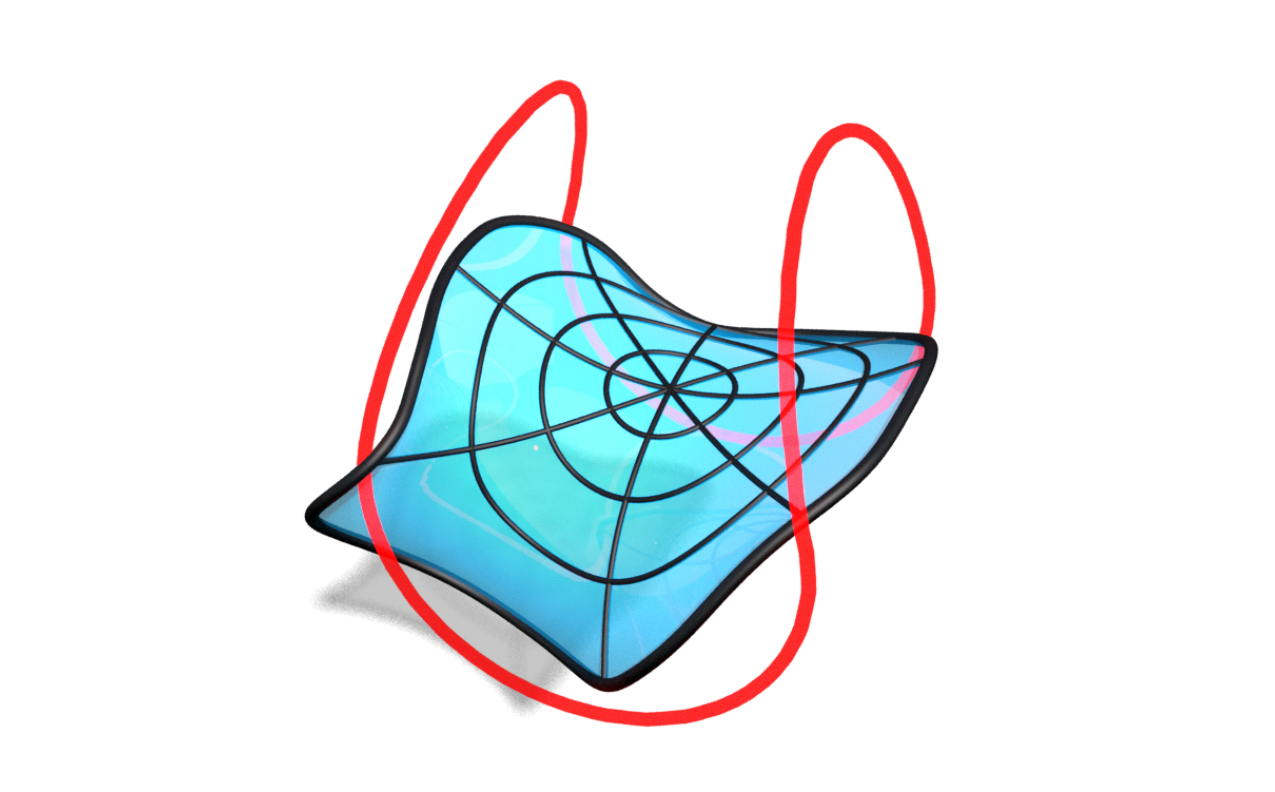} &
        \includegraphics[width=0.125\linewidth, trim=140 0 140 0, clip]{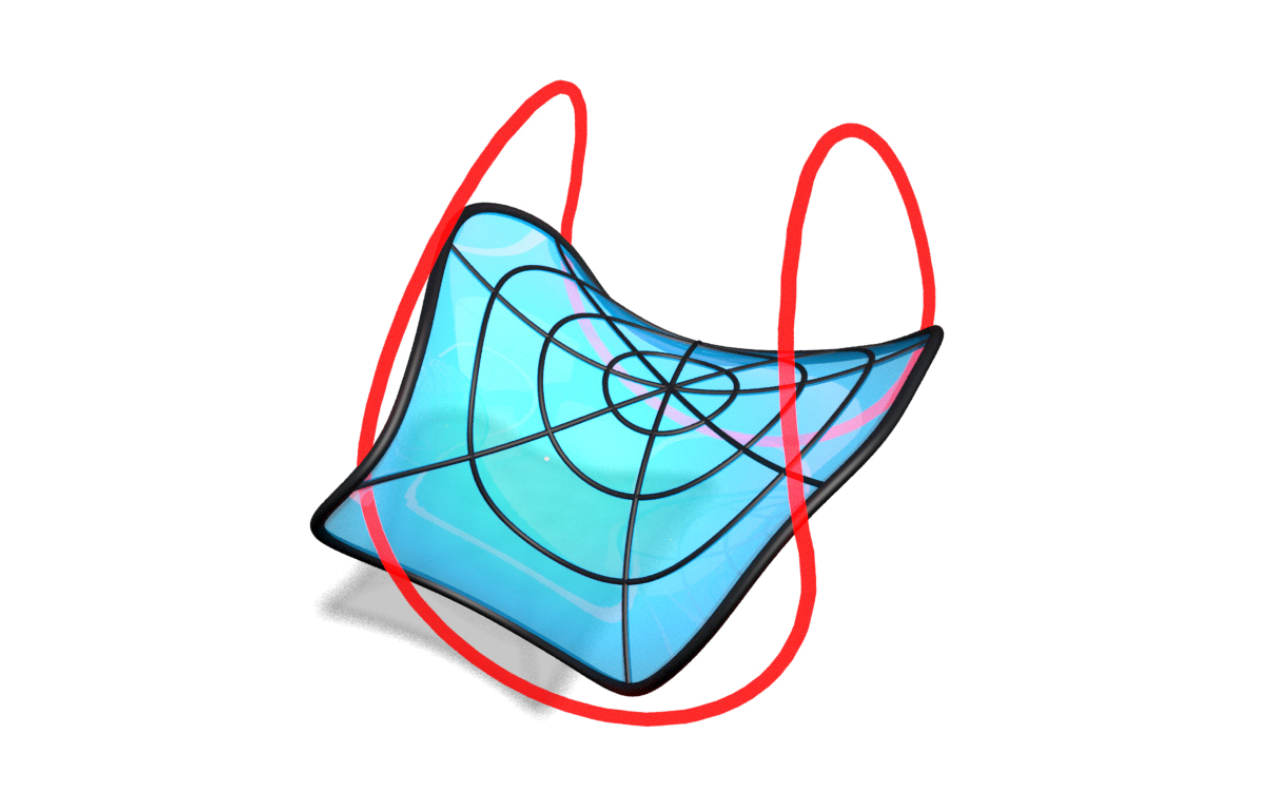} &
        \includegraphics[width=0.125\linewidth, trim=140 0 140 0, clip]{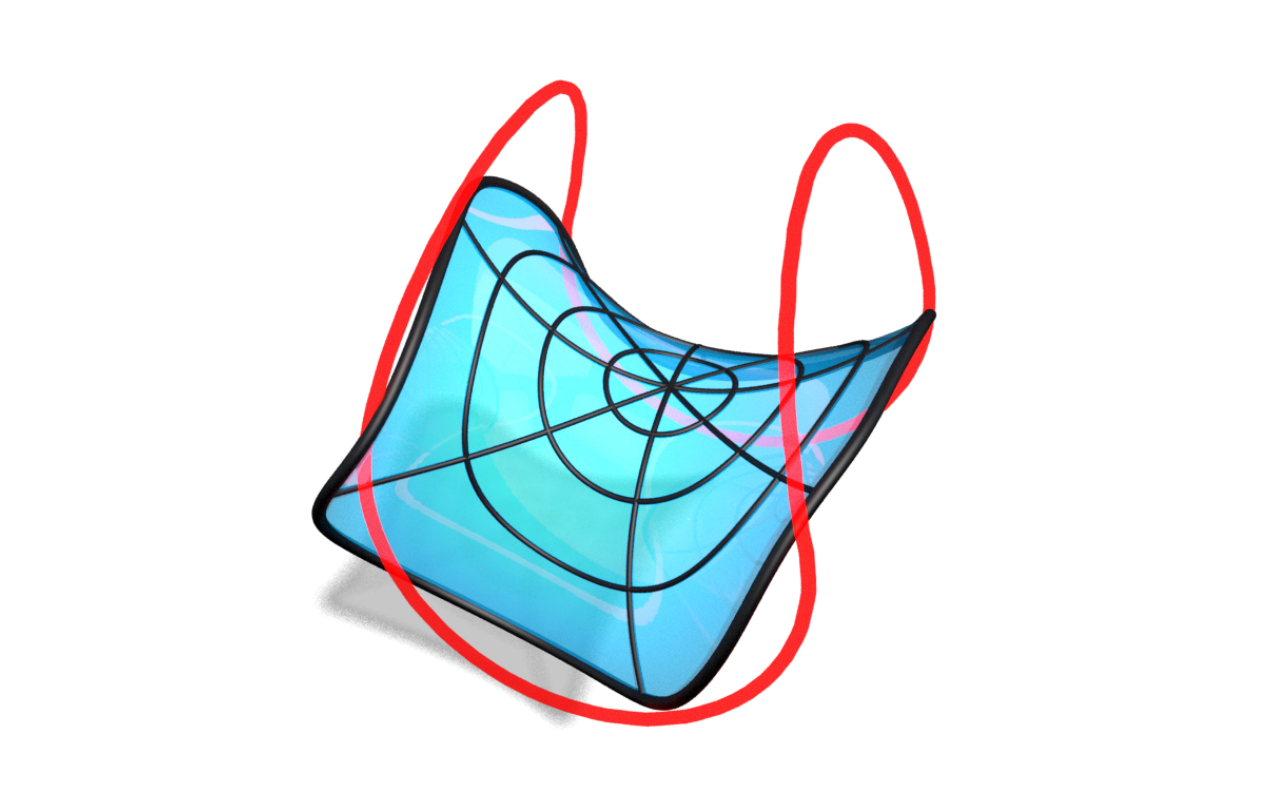} \\[0.3em]

        \makebox[0pt][r]{\raisebox{16pt}{\scriptsize $\lambda=0.5$}} &
        \includegraphics[width=0.125\linewidth, trim=140 0 140 0, clip]{images/0.pdf} &
        \includegraphics[width=0.125\linewidth, trim=140 0 140 0, clip]{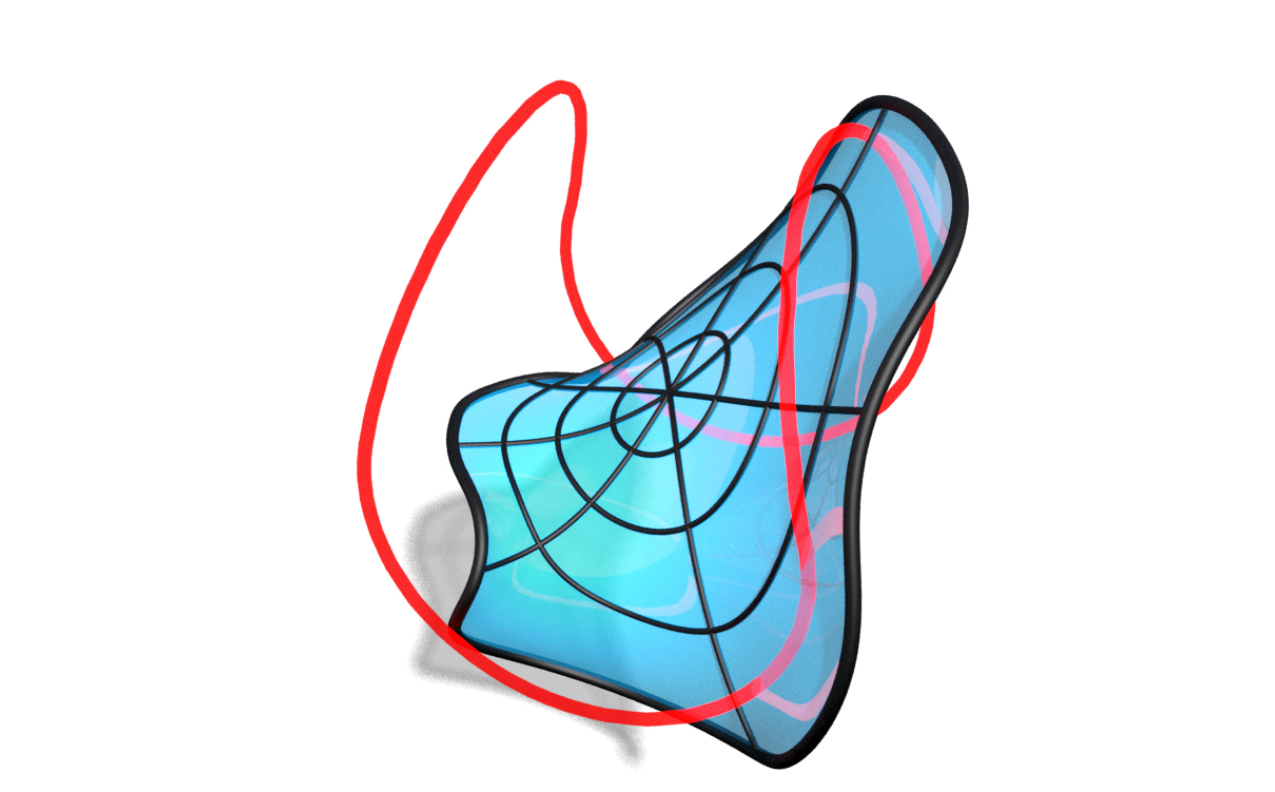} &
        \includegraphics[width=0.125\linewidth, trim=140 0 140 0, clip]{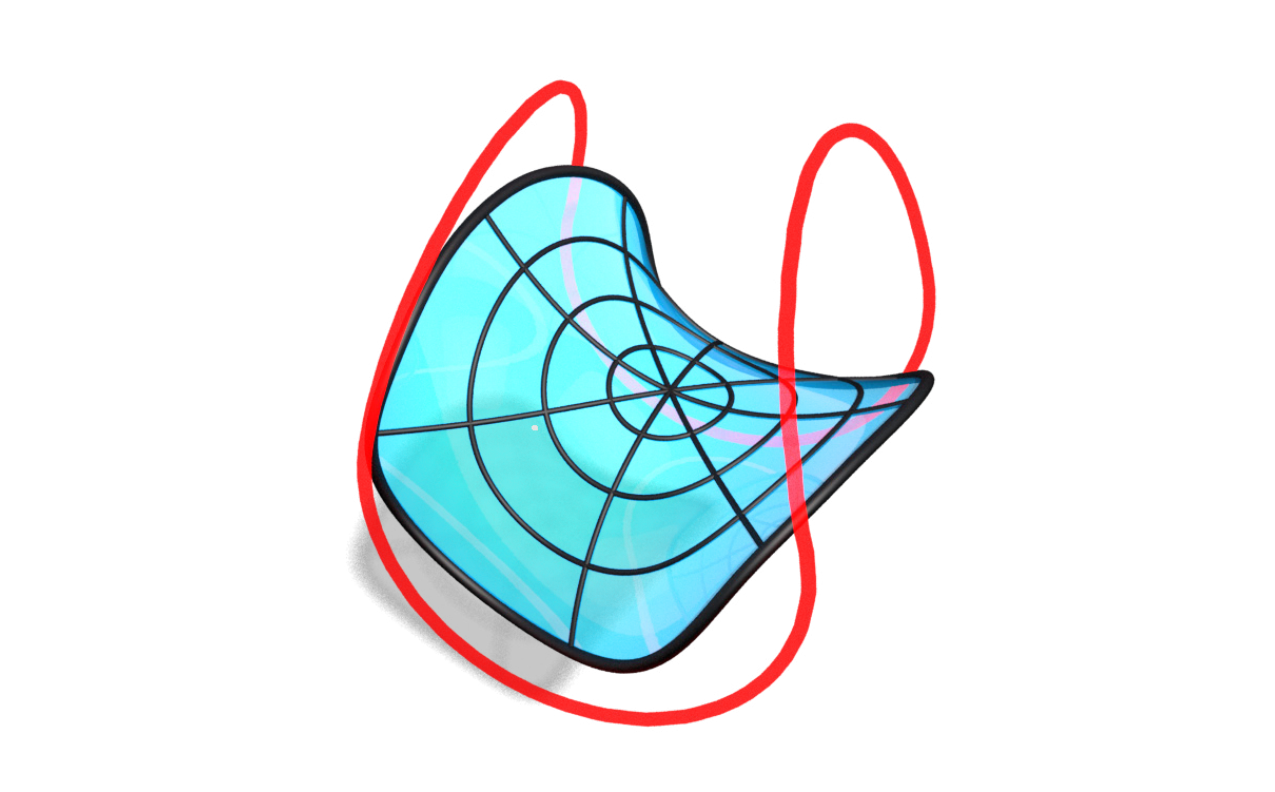} &
        \includegraphics[width=0.125\linewidth, trim=140 0 140 0, clip]{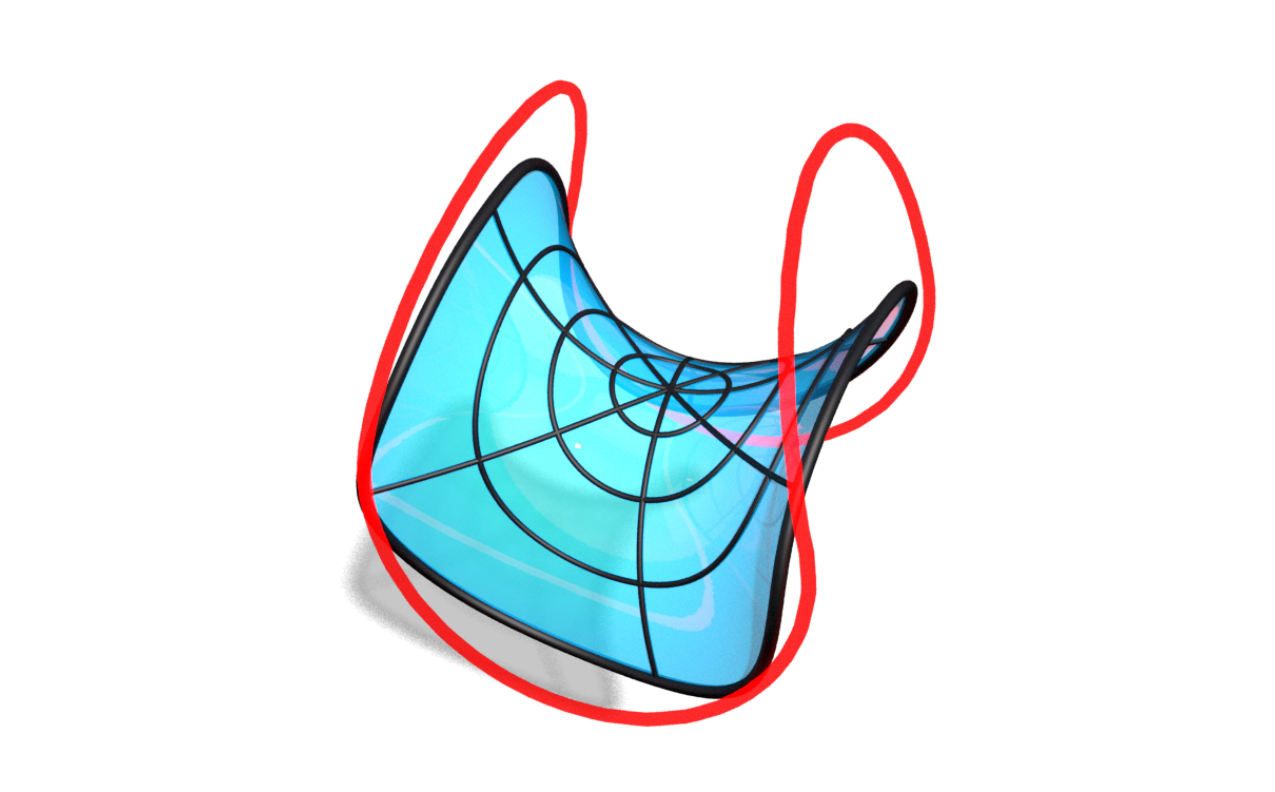} &
        \includegraphics[width=0.125\linewidth, trim=140 0 140 0, clip]{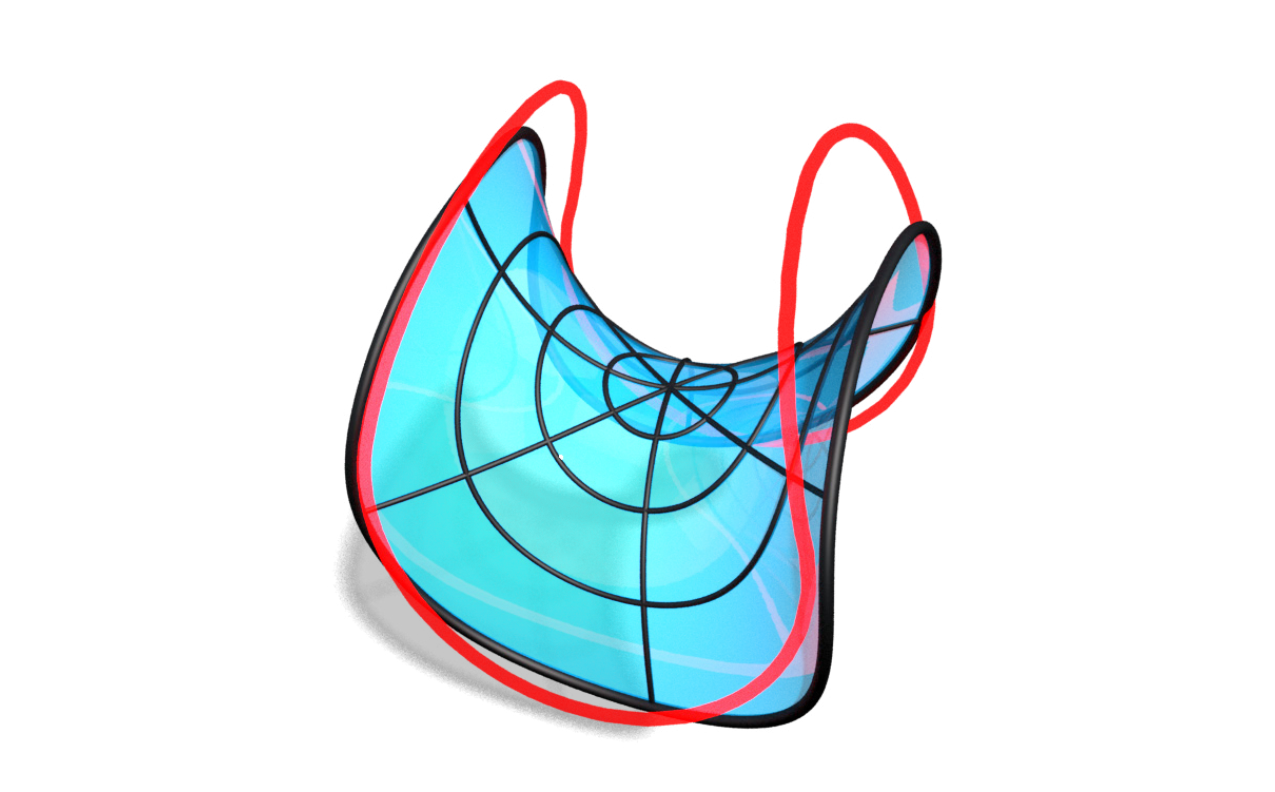} &
        \includegraphics[width=0.125\linewidth, trim=140 0 140 0, clip]{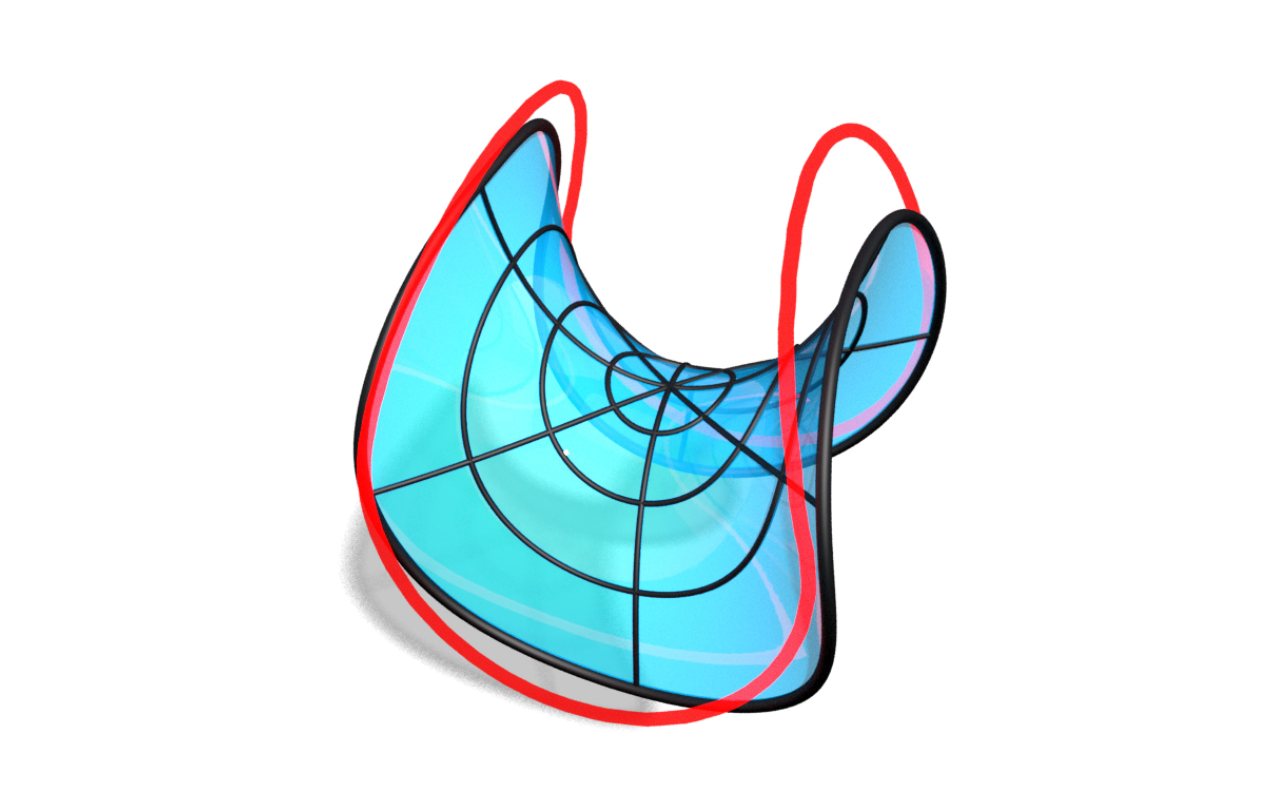}
    \end{tabular}
    }

    \caption{
    Optimization history with and without tangent fidelity to the boundary curve. 
    The case $\lambda=0$ disables the \(\cL_{\rm tangent}\) term, whereas $\lambda=0.5$ incorporates it. 
    Both experiments use the same initialization.
    }
    \label{fig:vel_comparison}
\end{figure}

\begin{figure}
    \centering
    \begin{subfigure}{\linewidth}
        \centering
        \includegraphics[width=0.8\linewidth]{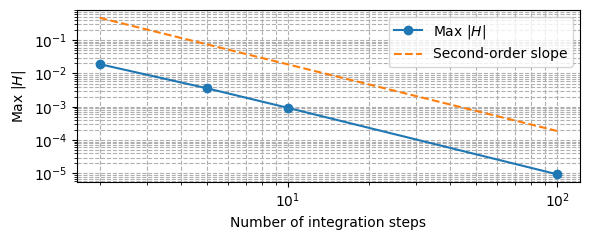}
        \caption{Surface of \Cref{fig:comparison}~(a)}
    \end{subfigure}
    \vspace{0.5em}
    \begin{subfigure}{\linewidth}
        \centering
        \includegraphics[width=0.8\linewidth]{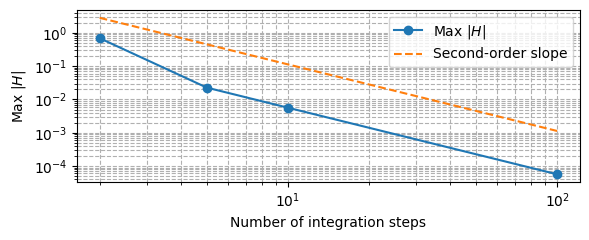}
        \caption{Surface of \Cref{fig:comparison2}~(a)}
    \end{subfigure}
    \vspace{0.5em}
    \begin{subfigure}{\linewidth}
        \centering
        \includegraphics[width=0.8\linewidth]{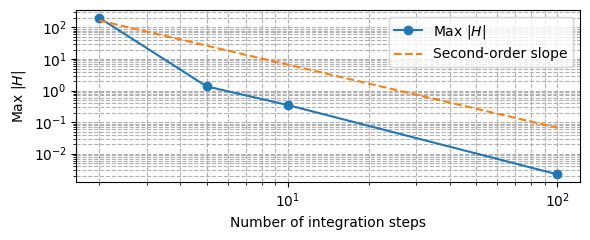}
        \caption{Surface of \Cref{fig:comparison_with_prescribed_surface}~(left)}
    \end{subfigure}
    \vspace{0.5em}
    \begin{subfigure}{\linewidth}
        \centering
        \includegraphics[width=0.8\linewidth]{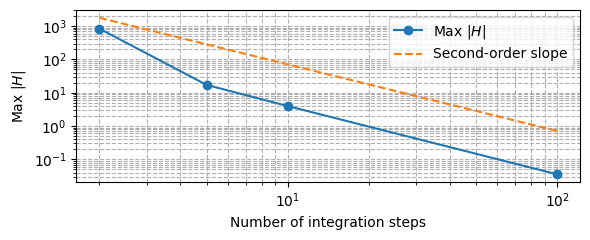}
        \caption{Surface of \Cref{fig:comparison_with_prescribed_surface}~(right)} 
    \end{subfigure}
    \caption{Log--log plots of the maximum absolute mean curvature versus the number of trapezoidal integration steps. The reference line indicates second-order convergence. Mean curvature is evaluated at points obtained by sampling a uniform $200\times200$ Cartesian grid over $[-1,1]^2$ and retaining only the points inside the unit disk.}
    \label{fig:Meancurvature_vs_intsteps}
\end{figure}

\clearpage

\appendix
    
\section{Conformal Condition for Harmonic Surfaces.}
\label{appendix:facts_about_harmonic_map}
Here we show that a harmonic surface \(\br = \Re\int_0^z\bphi(\zeta)\, d\zeta+\bc\) with a holomorphic \(\bphi\colon\DD\to\CC^3\) discussed in \eqref{eq:HarmonicSurfaceFromNullVelocity} is conformal if and only if \(\bphi\) is null (\ie\@ \(\bphi\) satisfies \eqref{eq:NullCondition}). 
Observe that $\phi_k=\frac{\partial{r_k}}{\partial x}-\ii\frac{\partial{r_k}}{\partial y}$. Thus
\begin{align*}
    \textstyle
    \sum_{k=1}^{3} {\phi}_k^{2} &=  
    \textstyle \sum_{k=1}^{3}{(\frac{\partial{r}_k}{\partial x})}^2
    - \sum_{k=1}^{3}{(\frac{\partial{r}_k}{\partial y})}^2-2 \ii \sum_{k=1}^{3}{\frac{\partial{r}_k}{\partial x}\frac{\partial{r}_k}{\partial y}}\\
    &= 
    \textstyle 
    |\frac{\partial{\br}}{\partial x}|^2-|\frac{\partial{\br}}{\partial y}|^2 -2\ii \langle\frac{\partial{\br}}{\partial x},\frac{\partial{\br}}{\partial y}\rangle.
\end{align*}
Hence \(\sum_{k=1}^3\phi_k^2 = 0\) if and only if \(|{\partial\br\over\partial x}|^2 = |{\partial\br\over\partial y}|^2\) and \(\langle {\partial\br\over\partial x},{\partial\br\over\partial y}\rangle = 0\), that is, \(\br\) is conformal.\qed

\section{Proof of \Cref{thm:RemovableBranch}}
\label{app:ProofOfRemovableBranch}
For a zero \(z_0\) of \(\bphi\) with order \(m\), we write \(\bphi(z) = (z-z_0)^m\bpsi(z)\) for some holomorphic \(\bpsi\) with \(\bpsi(z_0)\neq\bzero\).  
Let us study the type of branch point of the surface \(\br\) at \(z_0\) by the \({\frak q}\) value of the framed curve around \(z_0\).  Take the \(\epsilon\)-circle \(\gamma(\theta) = z_0 + \epsilon e^{\ii\theta}\), \(\theta\in[0,2\pi]\), for a small \(\epsilon>0\), and construct the basis vectors for the frame \((\bT(\theta),\bN(\theta),\bB(\theta))\in\SO(3)\) for each \(\theta\in[0,2\pi]\) by
\begin{align*}
\textstyle
    \bT = \textsc{Normalize}(d\br({d\gamma\over d\theta})),\;
    \bB = \textsc{Normalize}( d\br(-\ii{d\gamma\over d\theta})),\;\bN = \bB\times\bT.
\end{align*}
Here, the expression for \(\bB\) uses the fact that \(\br\) is conformal.  Note that 
\(d\br = \Re(d\brho) = \Re({\partial\brho\over\partial z}\, dz+{\partial\brho\over\partial\conj z}\, d\conj z) = \Re({\partial\brho\over\partial z}\, dz)
= \Re(\bphi\, dz)\) using the fact that \(\brho\) is holomorphic (\({\partial\over\partial \conj z}\brho = 0\)).  Hence, with \(dz({d\gamma\over d\theta}) = \epsilon\ii e^{\ii\theta}\), we have \(\bT(\theta) \propto \Re(\bphi(z_0 + \epsilon e^{\ii\theta})\epsilon \ii e^{\ii\theta}) \propto \Re( \ii e^{\ii (m+1)\theta}{\bpsi(\gamma(\theta))})\).
Similarly, \(\bB(\theta)\propto \Re(e^{\ii(m+1)\theta}\bpsi(\gamma(\theta)))\).
Since \(\bpsi\) is analytic and nonvanishing at \(z_0\), \(\bpsi(\gamma(\theta)) = \bpsi(z_0) + O(\epsilon)\).  Call \(\bpsi_0 = \bpsi(z_0)\in\CC^3\setminus\{\bzero\}\).
For simplicity let us define the (scaled) frame \(\tilde\bT(\theta)\coloneqq \Re(\ii e^{\ii (m+1)\theta}\bpsi_0)\) and \(\tilde\bB(\theta)\coloneqq\Re(e^{\ii(m+1)\theta}\bpsi_0)\).  Their normalizations represent the true frame vectors \(\bT,\bB\) up to \(O(\epsilon)\). In particular, for sufficiently small \(\epsilon\), the liftability of the frames (that determine the value of \({\frak q}(\gamma)\)) defined by \(\bT,\bB\) and by \(\tilde\bT,\tilde\bB\) are the same.

Now, observe that \(\tilde\bB(\theta),\tilde\bT(\theta)\) are in the span of  \(\tilde\bT(0) = -\Im(\bpsi_0)\), \(\tilde\bB(0) = \Re(\bpsi_0)\) for all \(\theta\).  Thus \(\tilde\bN\coloneqq {\tilde\bB\times\tilde\bT\over |\tilde\bB\times\tilde\bT|}\) is independent of \(\theta\), and 
\begin{align}
    \begin{bmatrix}
        \tilde\bB(\theta) & \tilde\bT(\theta)
    \end{bmatrix}
    =
    e^{\theta(m+1)[\tilde\bN\times{}]}\begin{bmatrix}
        \tilde\bB(0) & \tilde\bT(0)
    \end{bmatrix}
\end{align}
The lift of the rotation matrix \(e^{\theta(m+1)[\tilde\bN\times{}]}\) to quaternion is given by \(e^{\theta{m+1\over 2}\tilde\bN}\in\HH\), where the rotation axis \(\tilde\bN\) is regarded as an imaginary quaternion.  The liftability of the frame amounts to whether the quaternion \(e^{\theta{m+1\over 2}\tilde\bN}\) has consistent value at \(\theta = 0\) and \(2\pi\).  The frame is hence not liftable if and only if \(m\) is even.  Therefore, the value of \({\frak q}(\gamma)\) is zero if and only if \(m\) is even.  That is, the branch point at \(z_0\) is removable if and only if \(m\) is even.\qed

\section{Proof of \Cref{lem:even_order}}\label{app:ProofOfEvenOrderLemma}
\noindent $(\Rightarrow)$ Obvious.

\noindent $(\Leftarrow)$
\subsubsection*{\textbf{Claim 1.}} $\forall z_0\in \Omega$, $\exists u_U \in \mathcal{O}(U)$, where $z_0\in U \subset\Omega$, such that $u_U^2 = h$ on $U$. 
        \begin{proof}
        If $h=0$ at $z_0\in\Omega$, then locally we have $h(z) = (z-z_0)^{2k}g(z)$ on $\tilde{U}\ni z_0$ for some $g \in \mathcal{O}(\tilde{U})$ ($g(z_0)\neq0$). Pick a branch of logarithm and possibly shrink $\tilde{U}$ to $U\subset\tilde{U}$ such that $g(U)$ does not intersect the branch cut. Let $u_U(z):=(z-z_0)^{k}e^{\frac{1}{2}\log{g(z)}}$, then $u_U\in\mathcal{O}(U)$ and $u_U^2(z) = (z-z_0)^{2k}g(z) = h(z)$.
        \end{proof}

    \subsubsection*{\textbf{Claim 2.}} Suppose we have $u_U$ and $u_V$ on $U$ and $V$ respectively (with $U\cap V\neq\emptyset$ and $U$,$V$ are open disks) such that they satisfy the conditions that $u_U^2 = h$ on $U$ and that $u_V^2 = h$ on $V$, then either $u_U = u_V$ or $u_U = -u_V$ on $U\cap V$.
        \begin{proof}
        Let $r:=\frac{u_U}{u_V}$ on $U\cap V$. Then $r\in \mathcal{O}(U\cap V)$ and $r^2 = \frac{u_U^2}{u_V^2} = \frac{h}{h} =1$. Since $U\cap V$ is connected, we have $r = 1$ or $r = -1$. Therefore $u_U = u_V$ or $u_U = -u_V$.
        \end{proof}

    \subsubsection*{\textbf{Claim 3.}} For each $u_U$ defined above, it is defined on the arbitrary disk contained in $\Omega$ centered at $z_0$. 
        \begin{proof}
            Suppose not, then
            \[
            R:=\sup\{r\in\mathbb{R}^+|u_U\text{ can be extended to } B_r(z_0)\subset\Omega\}
            \]
            where $B_R(z_0)\subset \Omega$. Notice that on $B_R(z_0)$ we have $u_{B_R(z_0)}$ which satisfies $u_{B_R(z_0)}^2 = h$ and is indeed an holomorphic extension of $u_U$. 
            
            Then $\forall z\in C_R(z_0)$, $z\in\Omega$, and let $(u_{V_Z},V_z)$ be defined as above ($V_z$ being a disk) and we change the sign of $u_{V_Z}$ if necessary such that $u_{V_Z} = u_{B_R(z_0)}$. Then we define $\tilde{V}_z :=V_z\cup B_R$ and $u_{\tilde{V}_z}$ be the extension of $ u_{B_R(z_0)}$ to this set. Let $\tilde{V}:=\displaystyle \bigcup_{z \in U} \tilde{V}_z $, and $\forall \zeta\in \tilde{V}$ we have $\zeta\in \tilde{V}_z$ for some $z\in C_R(z_0)$, define $\tilde{u}(\zeta) := u_{\tilde{V}_z}(\zeta)$. This is well defined since for any other $\tilde{V}_{\hat{z}} \ni  \zeta$,  $u_{\tilde{V}_z} = u_{\tilde{V}_{\hat{z}}}$ on their overlap. So we have extended $u_U$ to a larger open set containing $B_R(z_0)$. By (Marsden 1.4.27) $B_R(z_0)$ is contained in a larger open disk in $\tilde{V}$ where $u_U$ is extended to, this contradicts the hypothesis.
        \end{proof}

    \subsubsection*{\textbf{Claim 4.}} $\exists u\in\mathcal{O}(\Omega)$ such that $u^2 = h$
        \begin{proof}

        Pick $z_0\in \Omega$. $\forall z\in \Omega$, for any continuous curve $\gamma:[0,1]\rightarrow \Omega$ joining $z_0$ and $z$, $\gamma([0,1])$ is compact and $\mathbb{C}\setminus\Omega$ is closed, by Distance Lemma (Marsden 1.4.21) $\exists \rho>0$ such that for each $t \in [0,1]$, $B_\rho(\gamma(t))\subset \Omega$. Also, since $\gamma$ is continuous, it is uniformly continuous on $[0,1]$, and hence $\exists\delta>0$ such that $|s-t|<\delta$ implies $|\gamma(s)-\gamma(t)|<\rho$. So we can choose a partition $\{0=t_0<...<t_n=1\}$ such that $|t_i - t_{i+1}|=\frac{\delta}{2}$, and for each $\gamma(t_i)$ we extend (or restrict) the associated function $u_{t_i}$ to $U_{t_i}:=B_{\rho}(\gamma(t_i))$ and name it $f_{t_i}$. This is a finite chain of balls, so we can switch the sign of $f_{t_{i+1}}$ starting from $i = 0$ such that $f_{t_{i}}$ and $f_{t_{i+1}}$ agree on $U_{t_i}\cap U_{t_{i+1}}$.

        We create a collection of pairs $(f_t,U_t),t\in[0,1]$ in the following way: 
        \begin{enumerate}
            \item for $t\in \{0=t_0<...<t_n=1\}$, we assign $(f_{t_i},U_{t_i})$. 
            \item $\forall t\in[0,1]\setminus\{0=t_0<...<t_n=1\}$, $t\in (t_j,t_{j+1})$ for some $j\in\{0,...,n-1\}$, we pick $U_t$ to be the open ball centered at $\gamma(t)$ such that it is contained in $U_{t_j}$.
        \end{enumerate}

        So $\forall t\in[0,1]$:
        \begin{enumerate}
            \item If $t\in \{0=t_0<...<t_n=1\}$, it is obvious that $\exists \epsilon >0$ such that $\forall t^\prime \in [0,1]$ such that $|t-t^\prime|<\epsilon$ implies $\gamma(t^\prime)\in U_t$ and the functions $f_{t}$ and $f_{t^\prime}$ coincide on the intersection $U_{t}\cap U_{t'}$.
            \item If $t\in[0,1]\setminus\{0=t_0<...<t_n=1\}$, it is also obvious that the same condition holds.
        \end{enumerate}

        Therefore, we have an analytic continuation along any curve. Combining the fact that $\Omega$ is simply connected, by \textbf{Monodromy Theorem} we have $u \in\mathcal{O}(\Omega)$ where $u(z):=u_\text{Analytic Continuation}(z)$ and that $u^2 = h$.
        \end{proof}

\section{Proof of \Cref{0-1 to 0-1 homeo_NN}}
\label{app:proofofhomeo}

\begin{lemma} \label{cts_approx_pcconst}
    $\forall\rho:[0,1]\to\mathbb{R}$ piecewise constant function which takes finitely many values, then $\exists\{g_i\}_{i\in I}$ which is a sequence of continuous functions $g_i:[0,1]\to\mathbb{R}$ such that $\{g_i\}_{i\in I} \rightarrow \rho$ strongly in $L^1([0,1])$.
\end{lemma}

\begin{proof}
    It suffices to prove for the case where $\rho$ is non-constant. Suppose $\rho(t) = s_k, t\in[t_k, t_{k+1}], k=0,...,N-1$. Let $J_k:= \lim_{t \to {t_k}^+} \rho(t) - \lim_{t \to {t_k}^-} \rho(t), k = 1, ..., N-1$.

    Pick $\delta>0$ such that $I_k:=[t_k - \delta, t_k+\delta],k = 1, ..., N-1$ are disjoint and contained in $[0,1]$, and that $\delta<\frac{\epsilon}{2\sum_{k=1}^{N-1}|J_k|}$.

    Define $g(s):=\begin{cases}
\text{$\rho(s)$}, & \text{$s\notin \bigcup_{k=1}^{N-1} I_k$}, \\
\text{$s_{k-1}+\frac{s-(t_k-\delta_k)}{2\delta_k}(s_k-s_{k-1})$}, & \text{$s\in I_k$ }.
\end{cases}$

Then 
\begin{equation}
    \begin{aligned}
        \|g-\rho\|_{L^1} &= \int_{[0,1]}|g-\rho|\,d\mu \\
        &= \int_0^1 |g(s)-\rho(s)|\,ds \\
        &= \sum_{k=1}^{N-1} \int_{t_k-\delta}^{t+\delta} |g(s)-\rho(s)|\,ds \\
        &\leq 2\delta\sum_{k=1}^{N-1}|J_k|\leq \epsilon
    \end{aligned}
\end{equation}

Define $\{g_i\}_{i\in I}$ such that $\|g_i-\rho\|_{L^1}<\frac{1}{i}$ using the construction above. Hence we have $\{g_i\}_{i\in I} \rightarrow \rho$ strongly in $L^1([0,1])$.
\end{proof}

\begin{corollary} \label{mlp_approx_pcconst}
    If $\rho$ is strictly positive. Then $\{g_i\}_{i\in I}$ can be replaced by $\{f_i\}_{i\in I}$ which is a sequence of MLPs with ReLU for hidden layers and exponential activation for the output layer.
\end{corollary}
\begin{proof}
    It is obvious from the construction of $\{g_i\}_{i\in I}$ in Lemma.~\ref{cts_approx_pcconst} that $g_i$ are also strictly positive.
    
    For each $g_i$, given $\epsilon>0$, by the Universal Approximation Theorem, $\exists \tilde{f}:\mathbb{R}\to\mathbb{R}$ MLP with ReLU activations in hidden layers and identity activation in the output layer such that $\left\| \tilde{f} - \log g_i\right\|_\infty < \epsilon$ on $[0,1]$.

    Let $f:=e^{\tilde{f}}$, then since $e^x$ is Lipschitz on $[(\inf_{x \in [0,1]} \log g_i(x)) -\epsilon, (\sup_{x \in [0,1]} \log g_i(x)) +\epsilon ]$ and we can choose $C=e^{\sup_{x \in [0,1]} \log g_i(x)) +1  }$ such that $C$ is invariant of $\epsilon$, we have $\forall x\in[0,1]$, $|f(x)-g_i(x)|=|e^{\tilde{f}(x)}-e^{\log g_i(x)}|=C\cdot|\tilde{f}(x)-\log g_i(x)|<C\cdot\epsilon$.

    Therefore $\forall \epsilon>0$ we can find $f$ such that $\|f-g_i\|_\infty<\epsilon$ and hence $\|f-g_i\|_{L^1}<\epsilon$ on $[0,1]$. Let $\epsilon=\frac{1}{i}$ and $f_i:=f$, it follows that $\|f_i-\rho\|_{L^1}\leq\|f_i-g_i\|_{L^1}+\|g_i-\rho\|_{L^1}<\frac{1}{i}+\frac{1}{i}=\frac{2}{i}$. Thus, we have $\{f_i\}_{i\in I} \rightarrow \rho$ strongly in $L^1([0,1])$.

\end{proof}

\begin{proof}
    Given $ \epsilon >0$. Since $\Phi$ is continuous, it is uniformly continuous on $[0,1]$. Therefore, we can find a $\delta >0$ such that $|\Phi(x)-\Phi(y)|<\epsilon$, $\forall |x-y|<\delta$. Pick a uniform partition of $[0,1]$: $0=t_{0}<t_{1}<\dots<t_N=1$ such that $|t_k-t_{k-1}|<\delta$. Let \[\Phi_N(t) := 
    \frac{t_{k+1} - t}{t_{k+1} - t_k} \Phi(t_k)
    + \frac{t - t_k}{t_{k+1} - t_k}  \Phi(t_{k+1}),\; t \in [t_k, t_{k+1}], \; k = 0,\dots,N-1.\] be the function of linear interpolation of values of $\Phi$ at $t_0, t_1,\dots,t_N$. Now observe that $\forall t\in[0,1]$, $t$ lies in $[t_k,t_{k+1}]$ for some $k$, we have 

    \begin{equation}
    \begin{aligned}
    |\Phi(t) - \Phi_N(t)|
    &= \left| \Phi(t) 
    - \frac{t_{k+1} - t}{t_{k+1} - t_k} \, \Phi(t_k)
    - \frac{t - t_k}{t_{k+1} - t_k} \, \Phi(t_{k+1}) \right| \\[6pt]
    &\leq 
    \left| \frac{t_{k+1} - t}{t_{k+1} - t_k} \, \Phi(t) 
    - \frac{t_{k+1} - t}{t_{k+1} - t_k} \, \Phi(t_k) \right| \\[6pt]
    &\quad +
    \left| \frac{t - t_k}{t_{k+1} - t_k} \, \Phi(t) 
    - \frac{t - t_k}{t_{k+1} - t_k} \, \Phi(t_{k+1}) \right| \\[6pt]
    &< \frac{t_{k+1} - t}{t_{k+1} - t_k} \, \epsilon
    + \frac{t - t_k}{t_{k+1} - t_k} \, \epsilon \\[6pt]
    &= \epsilon.
    \end{aligned}
    \end{equation}
    
    So $\|\Phi - \Phi_N\|_{\infty}<\epsilon$.
    
    The density function of $\Phi_N$ is obvious: Let $s_k:=\frac{\Phi(t_{k+1})-\Phi(t_{k})}{t_{k+1}-t_{k}}$ and $\rho_N(t):=s_k, t\in[t_k, t_{k+1}],k=0,...,N-1$. Observe that $\forall t\in [0,1], \Phi_N(t)=\int_0^t \rho_N(s)\,ds$.
    
    Let $m:=\min\limits_k s_k$, observe that $0<m\leq1$. Choose $\eta >0$ such that $0<\eta<\frac{m}{4}\epsilon$.
    
    Then by Corollary.~\ref{mlp_approx_pcconst}, we can find $\rho:[0,1]\to \mathbb{R}_{>0}$ which is an MLP with ReLU activations for hidden layers and exponential activation for the output layer such that $\|\rho-\rho_N\|_{L^1}<\eta$.

    Let $A(t):=\int_{0}^{t}\rho(s)ds$. Then we have $|A(t)-\Phi_N(t)|=|\int_0^t\rho(s)-\rho_N(s)\,ds|\leq\int_0^t|\rho(s)-\rho_N(s)|\,ds\leq\int_0^1|\rho(s)-\rho_N(s)|\,ds<\eta$. 

    Define $\Psi(t):=\frac{\int_0^t\rho(s)\,ds}{\int_0^1\rho(s)\,ds}$. Then $|\Psi(t)-\Phi_N(t)|=|\frac{\int_0^t\rho(s)\,ds}{A(1)}-\Phi_N(t)|=|\frac{A(t)-A(1)\Phi_N(t)}{A(1)}|=|\frac{A(t)-\Phi_N(t)-(A(1)-1)\Phi_N(t)}{A(1)}|\leq|\frac{A(t)-\Phi_N(t)}{A(1)}|+|\frac{A(1)-\Phi_N(1)}{A(1)}|\leq\frac{2\eta}{A(1)}\leq\frac{2\cdot\frac{m}{4}\cdot\epsilon}{\frac{m}{2}}$. (The last inequality holds because $|A(1)-1|<\eta<\frac{m}{2}$ and thus $A(1)>\frac{m}{2}=\epsilon$.)

    Finally we have $\|\Psi-\Phi\|_\infty\leq\|\Psi-\Phi_N\|_\infty+\|\Phi_N-\Phi\|_\infty\leq\epsilon+\epsilon=2\epsilon$.

\end{proof}

\end{document}